\documentclass{article}
\usepackage{amssymb}
\usepackage{amsmath}
\usepackage{amsthm}
\usepackage{times}
\usepackage{graphicx}
\usepackage{comment}
\newtheorem{theorem}{Theorem}
\newtheorem{definition}{Definition}
\newtheorem{corollary}{Corollary}
\newtheorem{proposition}{Proposition}
\newtheorem{lemma}{Lemma}
\newtheorem{observation}{Observation}
\newtheorem{claim}{Claim}
\newtheorem{intremark}{Intuitive remark}

\newcommand{\btb}{{\bf BTB}}

\newcommand{\tverts}{root vertices}
\newcommand{\tvertnot}{{\bf Roots}}
\newcommand{\leaf}{\ell}
\newcommand{\pex}{pseudoexpander}
\newcommand{\edgeProb}{n^{-4/5}}
\newcommand{\maxDegree}{3n^{1/5}}
\newcommand{\ppex}{pre-pseudoexpander}
\newcommand{\Ppex}{Pre-pseudoexpander}
\date{}
\title{Non-deterministic branching programs with logarithmic repetition cannot efficiently compute small monotone CNFs}
\author{Oded Lachish and Igor Razgon\\
        \small Department of Computer Science and Information Systems, Birkbeck, University of London\\
        \small      \{oded,igor\}@dcs.bbk.ac.uk
       }
\begin{document}
\maketitle
\begin{abstract}
In this paper we establish an exponential lower bound on the size of syntactic non-deterministic
read $d$-times branching programs for $d \leq \log n /10^5$ computing a class
of monotone CNFs with a linear number of clauses. 
This result provides the first separation
of classes NP and co-NP for syntactic branching programs with a logarithmic repetition 
and the first separation of syntactic non-deterministic branching programs  
with a logarithmic repetition from small monotone CNFs.
\end{abstract}
\section{Introduction}
We study here the complexity of syntactic non-deterministic branching programs ({\sc nbp}s) 
with bounded repetition $d$ (read-$d$-times {\sc nbp}s or simply $d$-{\sc nbp}s).
We prove an exponential
lower bound for $\log n/10^5$-{\sc nbp}s computing a class of functions that are expressible as
monotone {\sc cnf}s with a linear number of clauses. 
As the complement of a small {\sc cnf} is a small
{\sc dnf} having a linear size presentation as $1$-{\sc nbp}, our result separates 
{\sc np} and co-{\sc np} for branching programs with $\log n/10^5$ repetition.

The previous record for separating classes {\sc np} and co-{\sc np} for syntactic {\sc nbp}s was for 
$d$-{\sc nbp}s with $d$ up to $\Theta(\log n /\log \log n)$. It was achieved in 1994, by Jukna \cite{JuknaECCC94}.

To the best of our knowledge,
the paper of Jukna and Schnitger \cite{Jukna4clique} (Theorem 3.3.)
is the only result that separates small monotone {\sc cnf}s 
from $d$-{\sc nbp}s. It does so with $d=o(\log n/\log \log n)$. 
We note that they do not state this explicitly, but we can easily get to this conclusion because of the following.
Their lower bound is established 
for a function on graphs that is true if and only if the given
graph does not contain a $4$-clique. This function can be presented as a {\sc cnf}
with $6$ literals per clause, all of them negative. 
A simple reduction shows that if all these negative literals are replaced  by 
positive ones then the lower bound retains for the obtained monotone {\sc cnf}. 
We are also not aware of other results separating $d$-{\sc nbp}s, with $d>1$, from monotone functions or from 
non-monotone {\sc cnf}s.

Prior to these two results, in 1993, Borodin, Razborov and Smolensky 
proved an exponential
lower bound for $d$-{\sc nbp}s for $d$ up to $\Theta(\log n)$ computing a class of functions that is not known to be monotone or 
to have small $d$-{\sc nbp}  for its complement.
The results mentioned created  two natural open questions: 
(i) is there $d=\Theta(\log n)$ such that 
classes {\sc np}
and co-{\sc np} are different for syntactic branching programs with repetition up to $d$?; and
(ii) is it possible to separate  $d$-{\sc nbp} for \emph{some} $d=\Theta(\log n)$ 
from monotone functions or from {\sc cnf}s  or from monotone {\sc cnf}s? 
Our result resolves these problems by providing a positive answers to both.

As mentioned above,
besides \cite{Jukna4clique}, we are not aware of other super-polynomial lower bounds for $d$-{\sc nbp}s with
$d>1$ on {\sc cnf}s nor on monotone functions. 
However, such lower bounds exist for read-once branching
programs. For example, an exponential lower bound for deterministic read-once branching programs on
small monotone {\sc cnf}s was presented in \cite{GalCNF}. In \cite{RazgonLBAlgo}, the second author showed 
a parametrized $n^{\Omega(k)}$ for $1$-{\sc nbp}s lower bound computing monotone $2$-{\sc cnf}s whose underlying
graph has tree-width at most $k$. This lower bound can easily be converted into an exponential one by taking
the underlying graph to be an expander.

Currently, there are no super-polynomial lower bounds for syntactic branching
programs with super-logarithmic repetition, even deterministic ones. 
However, there is an exponential
lower bound for \emph{oblivious branching programs} with $o(\log^2 n)$ repetition \cite{OBDDLOG2}
(presented in a more general form of $o(n \log^2 n)$ trade-off). 
We note that oblivious branching programs are a special case of {\sc nbp}s 
with the same order of variable occurrences along every source-sink path. 

Researchers have also investigated branching programs for functions with non-Boolean
domains of variables. One purpose of considering this framework for {\sc nbp}s
is to obtain lower bounds for \emph{semantic} rather than syntactic {\sc nbp}s
that bound the number of variable occurrences on consistent paths only.
Such lower bounds are known for functions with non-Boolean domains \cite{Juknasemantic},
however, in the Boolean case, super-polynomial lower bounds have not yet
been established even for semantic $1$-{\sc nbp}s.

For a more detailed survey of research on branching programs, we refer the reader to
the monographs of Jukna \cite{Yukna} and Wegener \cite{WegBook}.

\paragraph{Overview of result:}
We construct a dedicated family of small monotone formulae. We then prove that 
for every member $\varphi$ of this family every  $\log n/10^5$-{\sc nbp} $Z$, where $n$ is the number of variables in $\varphi$, computing $\varphi$ has size exponential in $n$.

From a birds eye view the proof proceeds as follows.
Given an $\log n/10^5$-{\sc nbp} $Z$ computing a member $\varphi$ of the family, 
we show that the properties of $\varphi$ imply that every computational path $P$ of $Z$ contains 
a special set of $O(\log n)$ vertices, we call a \emph{determining set}. 
Then we show that $Z$ has an exponential number of distinct determining sets.
This concludes our result since there can be this many distinct determining sets 
only if the size of $Z$ is exponential in $n$.
We next describe how this is achieved.

We first reduce the problem of proving exponential lower bounds for $d$-{\sc nbp}s to the special case of proving exponential lower bounds for \emph{uniform} $d$-{\sc nbp}, 
which are {\sc nbp}s such that along every one of their source-sink paths, every variable occurs exactly $k$ times.
Then we look at an arbitrary computational path $P$ in $Z$ as a string where the variables on the path are the letters.
We use a structural theorem of Alon and Maass (Theorem 1.1. \cite{AMaass}), to deduce that there exists a set of $O(\log n)$ indices that split the string 
into a set of consecutive substrings (\emph{intervals}) each between two consecutive indices of this set.
This set of intervals has a special property that there exists two large sets of specific letters (variables), where by large we mean $n^\alpha$ for some $\alpha>0$, such that the first, which we call the \emph{odd set}, has letters that only appear in odd intervals, where the "odd" is according to their order, and the second, which we call the \emph{even set}, has letters that only appear in even intervals. 
The set of vertices on the path $P$ corresponding to the set of indices is the determining set $X$. 
We next explain how we use the special property of the intervals.

In order to utilize the properties of the odd and even sets we defined the dedicated family of {\sc cnf}s based on a special type of graphs we call \emph{\pex}s. 
This construction ensures that there is a large matching $M_X$ between members of the odd set and the even set, 
with each pair in the matching corresponding to a distinct clause in $\varphi$. 
Based on this matching we define a formula $\psi_X$, which is satisfied by the assignment of the computational path $P$. 
We construct a  probability space
on the set of satisfying assignment of $\varphi$, such that the probability that a formula $\psi_X$ is satisfied
is exponentially small in the size of $M_X$ and, since $M_X$ is large, is, in fact, exponentially small in the number of variables of $\varphi$.
Clearly the conjunction, over all the computational paths of $Z$, of formulae $\psi_X$ corresponding to the paths, is satisfied with probability $1$. 
Thus, there must be exponentially many formulae $\psi_X$ and in turn exponentially many distinct determining sets.

We remark, regarding the construction of {\sc cnf}s, that we first prove the 
existence of the considered class of {\sc cnf}s $\varphi$ non-constructively, using a
probabilistic method. In the Appendix, we show that there is a deterministic polynomial 
time procedure constructing a class of small monotone {\sc cnf}s $\varphi^*$ such 
that $\varphi$ can be obtained from $\varphi^*$ by a partial assignment to the latter. 
It can be observed that the exponential lower bound on $\log n/10^5$-{\sc nbp}s retains
for $\varphi^*$. That is, our result holds for a set of constructively created {\sc cnf}s.

\paragraph{Structure of the paper.}
The rest of the paper is structured as follows.
Section \ref{sec:prelim} contains all the preliminary definitions and notations.
Section \ref{sec:mainres} proves the main result with the proofs of the auxiliary results
postponed to the later sections.
In particular, Section \ref{sec:pe} establishes existence of  \pex s,
Section \ref{sec:cons} proves that each computational path of the considered branching programs
contains a determining set and Section \ref{sec:lbsize} proves 
the exponentially low probability claim described above.
The paper also has an appendix consisting of four sections.
In the Section A we provide a poly-size simulation of $d$-{\sc nbp}s
by uniform $d$-{\sc nbp}s. In Section B, we show
that the exponential lower bound for $O\log n/10^5$-{\sc nbp}s applies
to `constructively created' {\sc cnf}s (proving the theorem at the end of Section \ref{sec:mainres}).
In Section C we prove auxiliary lemmas for Section  \ref{sec:cons}.
In Section D we prove the validity of the definition of probability space used in Section \ref{sec:lbsize}.

\section{Preliminaries} \label{sec:prelim}
{\bf Sets of literals as assignments.}
In this paper by a \emph{set of literals} we mean one that does not
contain both an occurrence of a variable and its negation.
For a set $S$ of literals we denote by $Vars(S)$ the set of variables
whose literals occur in $S$ (the $Vars$ 
notation naturally generalizes to {\sc cnf}s, Boolean functions, and branching programs).
A set $S$ of literals
represents the truth assignment to $Vars(S)$ where variables occurring
positively in $S$ (i.e. whose literals in $S$ are positive) are assigned with $true$
and the variables occurring negatively are assigned with $false$.
For example, the assignment $\{x_1 \leftarrow true, x_2  \leftarrow true, x_3 \leftarrow false\}$
to variables $x_1,x_2,x_3$ is represented as $\{x_1,x_2,\neg x_3\}$.

{\bf Satisfying assignments.}
Let $\varphi$ be a {\sc cnf}. A set $S$ of literals \emph{satisfies} a clause $C$ of $\varphi$
if at least one literal of $C$ belongs to $S$. If all clauses of $\varphi$ are satisfied by $S$ then
$S$ \emph{satisfies} $\varphi$. If, in addition, $Vars(\varphi)=Vars(S)$ then we say that $S$ is a 
\emph{satisfying assignment} of $\varphi$. The notion of a satisfying assignment naturally extends
to Boolean functions $F$ meaning a truth assignment to $Vars(F)$ on which $F$ is true.

\begin{definition}[{\bf Non-deterministic branching programs and related notions}]
~~~~\newline
\begin{itemize}
\item A \emph{Non-deterministic branching program} ({\sc nbp}) $Z$ is a directed acyclic 
multigraph ({\sc dag}) with one source and one sink with some edges 
labelled with literals of variables. $Vars(Z)$ denotes the set of variables whose literals
occur on the edges of $Z$.
\item A directed source-sink path $P$ of $Z$ is a \emph{computational path} of $Z$ if opposite 
literals of the same variable do not occur as labels of edges of $P$.
We denote by $A(P)$ the set of literals labelling the edges of $P$.
We call $A(P)$ the assignment \emph{carried} by $P$. 
\item The function $F_Z$ \emph{computed} by $Z$ is a function on $Vars(Z)$. 
A set of literals
over $Vars(Z)$ is a satisfying assignment if and only if there is a computational path $P$ of $Z$
such that $A(P) \subseteq S$.
\item The size of $Z$, denoted by $|Z|$ is the number of its vertices. 
\end{itemize}
\end{definition}

\begin{definition}[{\bf Read $d$ times {\sc nbp}}]
A syntactic \emph{read-$d$-times} {\sc nbp} ($d$-{\sc nbp}) is an {\sc nbp} in which each variable
occurs at most $d$ times on each source-think path. A $d$-{\sc nbp} is \emph{uniform}
if each variable occurs \emph{exactly} $d$ times on each source-sink path.
\end{definition}

\begin{lemma} \label{lem:uniform}
For every  $d$-{\sc nbp} $Z$,
there exist a uniform $d$-{\sc nbp} computing
$F_Z$ and having size $O(|Z|^4d)$
\end{lemma}

The proof of Lemma \ref{lem:uniform} appears in the Appendix.

It follows from Lemma \ref{lem:uniform} that an exponential lower bound obtained
for a uniform $d$-{\sc nbp} applies to a $d$-{\sc nbp} in general.
Therefore, in the rest of the paper (except, obviously, for the proof of Lemma \ref{lem:uniform}),
we assume (without stating it explicitly) that all the considered $d$-{\sc nbp}s are uniform.
For the purpose of establishing the lower bound, the advantage of considering uniform {\sc nbp}s 
is that for any two paths of an {\sc nbp} having the same initial and final vertices,
the sets of variables labelling these paths are the same, as proved in the next lemma.

\begin{lemma} \label{lem:samevar}
Let $Z$ be a $d$-{\sc nbp}.
Let $P_1$ and $P_2$ be two paths of $Z$ having the same initial and final vertices.
Then a variable occurs on $P_1$ if and only if it occurs on $P_2$. 
\end{lemma}

{\bf Proof.}
Let $u$ and $v$ be the starting and ending vertices of $P_1$ and $P_2$.
Denote by $start$ and $end$ the source and sink vertices of $Z$, respectively.
Let $P_0$ be a $start \longrightarrow u$ path of $Z$ and $P_3$ be 
a $v \longrightarrow end$ path of $Z$.
Due to the acyclicity of $Z$, both $Q_1=P_0+P_1+P_3$ (the concatenation of $P_0,P_1,P_3$)
and $Q_2=P_0+P_2+P_3$ are source-sink paths of $Z$. 
Suppose there is a variable $x$ occurring on $P_1$ but not on $P_2$.
Then, to supply $d$ occurrences of $x$ on $Q_2$ required because $Z$ is uniform, 
$x$ occurs $d$ times in $P_1 \cup P_3$. It follows that on $Q_1$, $x$ occurs
at least $d+1$ times in contradiction to the definition of a $d$-{\sc nbp}.
$\blacksquare$


Now, we introduce terminology related to graphs and trees and we will need 
for definition of the {\sc cnf} class.

\begin{definition}
A rooted tree is called \emph{extended} if none of its leaves has a sibling.
\end{definition}

\begin{definition} [{\bf Binary tree based graphs} (\btb)] \label{def:btb}
A graph $H$ is a binary tree based graph if:
\begin{enumerate}
\item $H$ is an edge-disjoint union of extended trees ${\bf T}(H)=\{T_1, \dots, T_m\}$ with roots 
      $t_1, \dots, t_m$ which we call the \emph{\tverts} of $H$ and denote by $\tvertnot(H)$. The set of all leaves of 
      $T_1, \dots, T_m$ is denoted by ${\bf Leaves}(H)$.
\item Each vertex of $u \in {\bf Leaves}(H)$ is a leaf of exactly two trees of ${\bf T}(H)$.
\item Any two trees of ${\bf T}(H)$ have at most one vertex in common.
This common vertex is a leaf in both trees.
\end{enumerate}
In what follows we denote $|V(H)|$ by $n$ and $|{\bf Roots}(H)|$ by $m$.
\end{definition}

\begin{definition}[{\bf Adjacency of trees in} \btb] \label{def:adjtrees}
Let $H \in \btb$ with ${\bf T}(H)=\{T_1, \dots, T_m\}$
and ${\bf Roots}(H)=\{t_1, \dots, t_m\}$ where each $t_i$ is the root of $T_i$.
\begin{itemize}
\item $T_i$ and $T_j$ are adjacent if and only if they share a leaf in common.
The common leaf of $T_i$ and $T_j$ is denoted by $\leaf_{i,j}$.
\item The unique path connecting roots of $T_i$ and $T_j$ in $T_i \cup T_j$ is denoted by $P_{i,j}$.
\item The path connecting $t_i$ and $\leaf_{i,j}$ is in $T_i$ is denoted by $P_{i \rightarrow j}^{1/2}$
and the path connecting $t_j$ and $\leaf_{i,j}$ is in $T_j$ is denoted by $P_{j \rightarrow i}^{1/2}$.
(The $1/2$ in the subscript says that these paths are `halves' of $P_{i,j}$.)
\end{itemize}
\end{definition}

\begin{definition}[{\bf Pseudoedges and pseudomatchings}] \label{def:pseudoelem}
Let $H \in \btb$ and $U,V$ be two disjoint subsets of ${\bf Roots}(H)$.

\begin{itemize}
\item A \emph{pseudoedge} of $H$ is a pair $\{t_i,t_j\}$ of roots such that $T_i$
and $T_j$ are adjacent. The \emph{pseudodegree} of $t_i$ is the number
of pseudoedges containing $t_i$
\item A \emph{pseudomatching} of $H$ is a set of pseudoedges of $H$
that do not share common ends. 
\item A pseudomatching $M$ is \emph{between}
$U$ and $V$ if for each $e \in M$, $|e \cap U|=|e \cap V|=1$.
\end{itemize}
For a pseudomatching $M$, we denote $\bigcup_{e \in M} e$ by $\bigcup M$.
\end{definition}

\begin{definition} [{\bf Underlying graph for a graph in \btb}] \label{def:under}
The underlying graph $U(H)$ is a graph whose vertices are ${\bf Roots}(H)$ 
and two vertices in $U(H)$ are adjacent if and only
if their corresponding trees are adjacent in $H$. 
\end{definition}

\section{The main result} \label{sec:mainres}
{\bf Strategy of the proof.} 
  Fix $Z$ to be a $d$-{\sc nbp} that computes a $\varphi$.
  The goal is to show that $Z$ is exponentially large in the number $n$ of variables of $\varphi$.
  To do so we prove that every computational path $P$ of $Z$ has a set of variables $X$, of size logarithmic in $n$
  that has specific properties which enable us to prove the following: $Z$ contains an exponential in $n$ number of such distinct sets. 
   Since the size of these sets is logarithmic in $n$, the lower bound on the size of $Z$ easily follows.
   The specific properties of $X$ are listed next.
   \begin{enumerate}
   \item There is a special formula $\psi_X$ defined on the variables of $\varphi$
   such that for any path $Q$ containing $X$, the assignment $A(Q)$ carried by $Q$ satisfies $\psi_X$.
   \item There is a probability
   space on the set of satisfying assignments of $\varphi$, 
   so that the probability that an assignment drawn from the space satisfies $\psi_X$
   is exponentially small in $n$.
   \end{enumerate}
   
   Now, since every computational path contains a set $X$, every satisfying assignment $S$ of $\varphi$
   satisfies some $\psi_X$. 
   This, together with the last statement above, implies that the number of
   formulae $\psi_X$ must be exponentially large in $n$, and hence also the number of sets $X$.

The {\sc cnf} for which we prove the lower bound is based on a special class of graphs we call \pex s
and defined next and afterwards we define the actual {\sc cnfs}. 
Recall that $m$ denotes the number of roots of $H \in \btb$. 

\begin{definition}[{\bf Pseudoexpanders}]\label{def:pseudoexp}
$H \in \btb$ 
is a \emph{pseudoexpander} if the following two conditions hold.
\begin{enumerate}
\item \emph{Small height property.} The height (the largest number of vertices of a root-leaf path) 
of each $T_i$ is at most $(\log m)/4.9+3$.
\item \emph{Large pseudomatching property.} For any two disjoint subsets $U,V$ of $\tvertnot(H)$ of size at least $m^{0.999}$ each,
there is a pseudomatching of $H$ between $U$ and $V$ of size at least $m^{0.999}/3$.
\end{enumerate}
The the class of all pseudoexpanders is denoted by ${\bf PE}$.
\end{definition}

\begin{definition} \label{def:phipe}
For $H$, ${\bf T}(H)$ and $\tvertnot(H)$ as in Definition \ref{def:adjtrees},
a CNF $\phi_H$ is defined as follows.

The variables of $\phi_H$ are the vertices of $H$.
The clauses $C_{i,j}$ of $\phi_H$ are in bijective correspondence with the
pseudoedges of $H$ and the literals of $C_{i,j}$ are $V(P_{i,j})$.

The class $\{\phi_H|H \in {\bf PE}\}$ is denoted by ${\bf \Phi}({\bf PE})$.
\end{definition}


We are going to prove that $O(\log n)$-{\sc nbp}s require an exponential
size to compute the class ${\bf \Phi}({\bf PE})$.
Note that Definition \ref{def:pseudoexp} does not obviously imply that even one
pseudoexpander exists, while, for the purpose of the proof we need an infinite number
of them. This is established in the following theorem.

\begin{theorem} \label{peinfprel}
For each sufficiently large $m$ there is a \pex with $m$ roots and max-degree at most 
$m^{1/4.9}$. That is, there are infinitely many $n$ for which there is a \pex with $n$ 
vertices and the pseudodegree of each root being at most the number of roots to the power
of ${1/4.9}$. 
\end{theorem}

Theorem \ref{peinfprel} is proved in Section \ref{sec:pe} using probabilistic method.

\begin{corollary} \label{peinf}
The elements of ${\bf \Phi}({\bf PE})$ are monotone CNFs with linear number of clauses
and there is an unbounded number of them.
In particular, there is an unbounded number of $\phi_H$ where $H$ is a \pex with $m$
roots such that the pseudodegree of each root is at most $m^{1/4.9}$.
\end{corollary}

We next define the formulas that will be associated with the computational paths as described in the overall intuition.

\begin{definition} [{\bf Matching CNF}] \label{def:mcnf}
Let $M$ be a pseudomatching of $H \in \btb$. 
A matching CNF w.r.t. $M$ consists of clauses $C^{1/2}_{i,j}$ for each $\{t_i,t_j\} \in M$,
where $C^{1/2}_{i,j} \in \{V(P^{1/2}_{i \rightarrow j}),V(P^{1/2}_{j \rightarrow i})\}$. 

The set of all matching CNFs w.r.t. $M$ is denoted by ${\bf CNF}(M)$.
\end{definition}


\begin{definition} [{\bf Matching OR-CNF}] \label{def:morcnf}
An matching OR-CNF w.r.t. a pseudomatching $M$ of $H \in \btb$ has the
form $Conj(S_1) \wedge \phi_1 \vee \dots \vee Conj(S_q) \wedge \phi_q$,
where $\{S_1, \dots, S_q\}$ is the family of all sets of literals 
over $V(H) \setminus \bigcup M$, $Conj(S_i)$ is the conjunction
of literals of $S_i$
and $\phi_i \in {\bf CNF}(M)$ for each $1 \leq i \leq q$.
The set of all matching OR-CNFs w.r.t. $M$ is denoted by ${\bf ORCNF}(M)$.
\end{definition}

We now define the special set of variables that appears in every computational path.
Recall first that $n$ and $m$ denote the number of vertices of roots of $H$, respectively.

\begin{definition} [{\bf Determining set}] \label{def:detset}
Let $Z$ be an NBP implementing $\phi_H$ for $H \in \btb$.
A set $X \subseteq V(Z)$ is called an $(a,b)$-\emph{determining set}
if $|X| \leq a$ and there are a pseudomatching $M$ of $H$ of size at
least $b$ and $\psi_X \in {\bf ORCNF}(M)$ such that for any computational
path $Q$ of $Z$ containing $X$, $A(Q)$ satisfies $\psi_X$.
We call $\psi_X$ a $b$-\emph{witnessing formula} for $X$. 
\end{definition}


\begin{theorem} \label{lbcons}
Let $Z$ be an $d$-NBP implementing $\phi(H)$ for $H \in {\bf PE}$ 
with $d \leq \log n/10^5$.
Then each computational path of $Z$ contains a $(\log n/4000,m^{0.999}/3)$-determining set.
\end{theorem}

Theorems \ref{lbcons} is proved in Section \ref{sec:cons}. 
The following theorem is based on the low probability of satisfaction of $\psi_X$
for a $(\log n/4000,m^{0.999}/3)$-determining set $X$ 
as described in the strategy of the proof. 


\begin{theorem} \label{lbsize}
There are constants $\alpha>1$ and $0<\gamma<1$ such that the following holds. 
Let $H$ be a \pex\ with a sufficiently large number of vertices
and let $\psi_1 \dots, \psi_q$ be OR-CNFs for which the following properties hold.
\begin{itemize}
\item For each $\psi_i$, there is a pseudomatching $M_i$ of $H$ of size at least $m^{0.999}/3$
such that $\psi_i \in {\bf ORCNF}(M_i)$. 
\item Every satisfying assignment of $\phi_H$ satisfies $\psi_1, \dots, \psi_q$. 
\end{itemize}
Then $q \geq \alpha^{n^{\gamma}}$.
\end{theorem}

Theorem \ref{lbsize} is proved in Section \ref{sec:lbsize}.

\begin{theorem} \label{lbmain}
There exist constants $\alpha>1$ and $0<\beta<1$ such that 
the following holds:
if $H$ is a sufficiently large \pex\  and $Z$ a $d$-{\sc nbp} 
computing $\phi_H$ with $d \leq  \log n/10^5$,
then $|Z| \geq \alpha^{n^{\beta}}$.
\end{theorem}

{\bf Proof.}
By Theorem \ref{lbcons}, each computational path $P$ contains
$X(P)$ which is a $(\log n/4000,m^{0.999}/3)$-determining set.
Let $\{X_1, \dots, X_q\}$ be all such sets $X(P)$
and let $\psi_1, \dots, \psi_q$ be their respective $m^{0.999}/3$-witnessing
formulas. That is there are pseudomatchings $M_1, \dots M_q$ of $H$
of size at least $m^{0.999}/3$ each such that $\psi_i \in {\bf ORCNF}(M_i)$
for $1 \leq i \leq q$. 

Next, we show that the disjunction of $\psi_i$ is satisfied
by every satisfying assignment of $\phi_H$.
Indeed, let $S$ be a satisfying assignment of $\phi_H$.
Then there is a computational path $P$ with $S=A(P)$.
By the definition of $\{X_1, \dots, X_q\}$, there is some $X_i$ contained in $P$.
It follows from Definition \ref{def:detset} that $S$ satisfies $\psi_i$
and hence the disjunction of $\psi_1, \dots, \psi_q$.

Thus both conditions of Theorem \ref{lbsize} are satisfied and it
follows that $q \geq \alpha^{n^{\gamma}}$.
for some $\alpha>1$ and $0<\gamma<1$.
As each $X_i$ is of size at most $\log n/4000$, $q \leq |Z|^{\log n}$
and thus $|Z|^{\log n} \geq \alpha^{n^{\gamma}}$.
Therefore, $|Z| \geq \alpha^{n^{\beta}}$ where $\beta=0.99*\gamma$.
$\blacksquare$

The class of {\sc cnf}s considered in Theorem \ref{lbmain} is non-constructive,
its existence is established using the probabilistic method.
In the following theorem, we prove existence of a poly-time deterministic
procedure creating a class of monotone small {\sc cnf}s such that 
for each sufficiently large {\sc cnf} $\varphi$ of this class there is a partial
assignment transforming $\varphi$ into $\phi_H$ for some \pex\ $H$ of essentially
the same size. Thus the lower bound of Theorem \ref{lbmain} applies to this
constructively created class of {\sc cnf}s.

\begin{theorem} \label{nbpfinal}
There is a class of {\sc cnf}s produced by a polynomial time deterministic procedure
that takes an exponential size to compute by $d$-{\sc nbp} for $d \leq  \log n/10^5$
\end{theorem}

The proof of Theorem \ref{nbpfinal} appears in the Appendix.

\section{Proof of Theorem \ref{peinfprel}} \label{sec:pe}
In order to prove Theorem \ref{peinfprel}, 
we introduce a new type of graphs we call \emph{\ppex}, which we define next.
Afterwards we show that for each sufficiently large $m$ there is a \ppex\ with $m$ vertices.
Then we prove Theorem \ref{peinfprel} by showing that for a sufficiently large $m$ there is a \ppex\ with $m$
vertices that is the underlying graph of \pex\ with $m$ roots.

\begin{definition}\textbf{[\Ppex]}\label{def:PPE}
An $m$ vertex graph $G=(V,E)$ is a \emph{\ppex} if it satisfies the following:
\begin{enumerate}
\item\label{cond:PERank}
 the maximal degree of a vertex in $G$ does not exceed $m^{1/4.9}$, and
 \item\label{cond:PEMatching}
 for every disjoint $V_1,V_2 \subseteq V$ such that $|V_1|,|V_2| \geq n^{0.999}$, 
there exists a matching in $G$, of size at least $m^{0.999}/2$, between $V_1$ and $V_2$.
\end{enumerate}
\end{definition}

The proof that for each sufficiently large $m$ there is a \ppex\ with $m$ vertices 
consists of two parts: in the first part, we show that the existence of a \ppex\ 
follows from the existence of a graph with the small degree property (Condition~\ref{cond:PERank}) 
and a relaxation of the second property (Condition~\ref{cond:PEMatching}) from matching to single edges; 
and in the second part we use the probabilistic method to prove that such graphs with $m$
vertices exist for each sufficiently large $m$. 

\begin{lemma}\label{lem:PEreduction}
Let $G=(V,E)$ be a graph of $m$ vertices.
If for every  disjoint $U_1,U_2 \subseteq V$ such that $|U_1|,|U_2| \geq m^{0.999}/2$,
there exists an edge $\{u,v\}$ such that $u\in U_1$ and $v\in U_2$, then $G$ satisfies 
Condition~\ref{cond:PEMatching} of the definition of a \ppex.
\end{lemma}

\begin{proof}
Assume for the sake of contradiction that there exist disjoint subsets 
$U_1,U_2\subseteq V$, each of size at least $m^{0.999}$,
such that every matching between them is of size less than $m^{0.999}/2$.
Let $W$ be the edges of a maximum matching between $U_1$ and $U_2$.
Since the number of edges in $W$ is less than $m^{0.999}/2$,
$|U_i\setminus V(W)|>m^{0.999}-m^{0.999}/2=m^{0.999}/2$ for $i \in \{1,2\}$.
Consequently, there is an edge $\{u,v\}$ such that $u\in U_1\setminus V(W)$ and $v\in U_2\setminus V(W)$.
Thus, $W\cup \{u,v\}$ is also a matching between $U_1$ and $U_2$.
Since obviously, $W\cup \{u,v\}$ is larger than $W$ we get a contradiction to 
$W$ being a maximum matching between $U_1$ and $U_2$. 
\end{proof}

\begin{lemma}\label{lemma:AlmostPEExist}
There exists a constant $c$ such that, for every $m$ larger than $c$, there exists a $m$ vertex graph $G=(V,E)$ that satisfies the following two conditions:   
\begin{enumerate}
\item the maximal degree of a vertex in $G$ does not exceed $m^{1/4.9}$, and
\item for every  disjoint $U_1,U_2 \subseteq V$ such that $|U_1|,|U_2| \geq m^{0.999}/2$,
there exists an edge $\{u,v\}$ such that $u\in U_1$ and $v\in U_2$.
\end{enumerate}
\end{lemma}
\begin{proof}
We construct a random graph and prove that with strictly positive probability it has the property that, for every
 pair of subsets $U_1,U_2 \subseteq V$ such that $|U_1|,|U_2|= \lceil m^{0.999}/2 \rceil$, there exists an edge 
$\{u,v\}$ such that $u\in U_1$ and $v\in U_2$. Thus, by the probabilistic method a graph with such a property exists.  
The statement of the lemma follows, because every set of size 
at least $m^{0.999}/2$ is a superset of a set of size exactly $\lceil m^{0.999}/2\rceil$.

Let $G = (V,E)$ be a random graph of $n$ vertices such that, for every distinct $u,v\in V$, the edge $\{u,v\}$ is in $E$ independently with probability $\edgeProb$.
Let $|U_1|,|U_2| = \lceil m^{0.999}/2\rceil$. 

There are $|U_1|*|U_2| \geq m^{1.998}/9$ pairs $\{u,v\}$ such that $u\in U_1$ and $v\in U_2$ and hence, with probability at least 
$$(1-\edgeProb)^{m^{1.998}/9} \leq e^{-\edgeProb*m^{1.998}/9}=e^{-m^{1.198}/9} \leq e^{-m^{1.197}},$$ 
there does not exist a single edge in $E$ between a vertex in $U_1$ and a vertex in $U_2$.
The number of pairs of disjoint subsets of $V$ each of size exactly $\lceil m^{0.999}/2\rceil$ is at most:
$${m \choose m^{0.999}}*2^{m^{0.999}} \leq
m^{2*m^{0.999}}*2^{m^{0.999}} \leq 2^{\log m*2*m^{0.999}}*2^{m^{0.999}} \leq 2^{m^{0.9991}} \leq
e^{m^{0.9991}}.$$

Thus, by the union bound, with probability at least $1-e^{m^{0.9991}}\cdot e^{-m^{1.197}}>2/3$, for every
pair of subsets $U_1,U_2 \subseteq V$ such that $|U_1|,|U_2| = \lceil m^{0.999}/2\rceil$, 
there exists an edge $\{u,v\}$ such that $u\in U_1$ and $v\in U_2$.
 
Let $v$ be an arbitrary vertex in $G$.
By the Chernoff's bound, the probability that the degree of $v$ is greater than $\maxDegree$ is at most $2e^{-m^{1/5}}$.
Thus, by the union bound, with probability at least $1-m*2e^{-m^{1/5}}>2/3$, for every $u\in V$, the degree of $U$ is at most $\maxDegree$.

Finally, again by the union bound, with probability strictly greater than $1/3$, the random graph we defined has both properties required in the statement of the lemma.
\end{proof}

{\bf Proof of Theorem \ref{peinfprel}}
According to Lemma \ref{lemma:AlmostPEExist},
for each sufficiently large $m$, there is a graph $G'_m$
with $m$ vertices, max-degree $m^{1/4.9}$ and with an edge between
any two subsets of its vertices of size at least $m^{0.999}/2$.
According to Lemma \ref{lem:PEreduction}, $G'_m$ is in fact \ppex.

We are going to construct a graph $H_m \in \btb$ with $m$ roots such
that $G'_m$ is isomorphic to $U(H_m)$, the underlying graph of $H_m$
and with the height of each tree of ${\bf T}(H_m)$ being at most
$\log m/4.9+3$. The construction is as follows. 
Enumerate the vertices of $G'_m$ by $1, \dots, m$.
Fix extended trees $T_1, \dots, T_m$ such that the number of leaves of $T_i$
is exactly as $d_i$, the degree of vertex $i$ in $G'_m$ and the height of $T'_i$ is at
most $\log m/4.9+3$. It is not hard to see that such trees exist.
Indeed, take a complete binary tree $T'$ with at most $2*m^{1/4.9}$ leaves,
fix an arbitrary set $L$ of $d_i \leq m^{1/4.9}$ leaves and let 
$T''$ be the rooted tree obtained by the union of all root-leaf paths of $T'$ ending at $L$.
The height of $T'_i$ is at most $\log m/4.9+2$. By adding one leaf to each leaf
of $T''$ we obtain an extended tree with $d_i$ leaves and height at most $\log m/4.9+3$.
Next for each $1 \leq i \leq m$ mark the leaves of $T_i$ with the numbers assigned to
the neighbour of vertex $i$ in $G'_m$ (each leaf is assigned with a unique number).
Then for those pairs $\{i,j\}$ where $T_i$ has a leaf $\ell_1$ marked with $i$ and
$T_j$ has a leaf $\ell_2$ marked with $j$, identify $\ell_1$ and $\ell_2$.
Let $H_m$ be the resulting graph. A direct inspection shows that $U(H_m)$ is isomorphic
to $G'_m$.

We next show that $H_m$ will satisfy the large pseudomatching
property. Indeed, let $U_1,U_2$ be two disjoint subsets of roots of
$H_m$ of size at least $m^{0.999}$ each. By the property of \ppex s,
there is a matching of size at least $m^{0.999}/2 \geq m^{0.999}/3$ between
$U_1$ and $U_2$. As in $G'_m$ there is an edge between $u \in U_1$ and $v \in U_2$
if and only if in $H_m$ there is a pseudoedge between these vertices in $H_m$,
there is a pseudomatching of size at least $m^{0.999}/3$ between $U_1$ and $U_2$ in $H_m$.
We conclude that $H_m$ is a \pex\ with maximal pseudodegree at most $m^{1/4.9}$,
thus implying the theorem.
$\blacksquare$

\section{Proof of Theorem \ref{lbcons}} \label{sec:cons}
In this section, we prove that each computational path 
of a $d$-{\sc nbp} $Z$ with $d \leq \log n/10^5$ computing
$\phi_H$ for $H \in {\bf PE}$,
contains a $(\log n/4000,m^{0.999}/3)$-determining set $X$.
We first show, using the structural theorem of Alon and Maass  
(Theorem 1.1. of \cite{AMaass}), that on each computational path $P$,
there is a set $X$ of size at most $\log n/4000$ that \emph{separates}
two large disjoint sets $Y_1$ and $Y_2$ of roots of $H$ in the following
sense. We call $Y_1$ the \emph{odd} set and
$Y_2$ the \emph{even} set. The set $X$ naturally partitions $P$
into subpaths enumerated along $P$ so that the elements
of the $Y_1$ occur as labels only on edges of odd numbered 
subpaths and elements of $Y_2$ occur only on even numbered subpaths.
We then show that $X$ is in fact a $(\log n/4000,m^{0.999}/3)$-determining set.

That is, we show the existence of a formula 
$\psi_X \in {\bf ORCNF}(M)$, where $M$ is a pseudomatching of size at 
least $m^{0.999}/3$ such that every assignment carried by a computational
path containing $X$ satisfies $\psi_X$.

We fix a pseudomatching $M$ between $Y_1$ and $Y_2$ such 
that $|M| \geq m^{0.999}/3$. The existence of $M$ is guaranteed by the 
large pseudomatching property of pseudoexpanders.
We consider an arbitrary assignment $S$
to all the variables of $\phi_H$ except $\bigcup M$. We show that for
each $\{t_i,t_j \} \in M$, there is a `half clause' 
$C'_{i,j} \in \{V(P^{1/2}_{i \rightarrow j}), V(P^{1/2}_{j \rightarrow i})\}$
satisfied by the assignment carried by every computational path $P$
containing $X$ with $S \subseteq A(P)$. We prove this by showing that otherwise
there are computational paths $Q'$ and $Q''$ both containing $X$ and with 
$S \subseteq A(Q') \cap A(Q'')$ such that $Q'$ falsifies one of these clauses
and $Q''$ falsifies the other one. Then, because $Z$ is a uniform $d$-{\sc nbp}, 
$Q'$ and $Q''$ can be `combined' into a computational path $Q^*$ 
falsifying $C_{i,j}$, in contradiction to the definition of $Z$. 
It follows that every assignment carried by a computational path $P$
containing $X$ with $S \subseteq A(P)$ satisfies $\varphi \in {\bf CNF}(M)$
consisting of the half clauses $C'_{i,j}$ for $\{t_i,t_j\} \in M$.

In the final stage of the proof, we set
$\psi_X=Conj(S_1) \wedge \varphi_1 \vee \dots \vee Conj(S_q) \wedge \varphi_q$,
where $S_1, \dots, S_q$ are all possible sets of literals assigning all the variables
except $\bigcup M$ and $\varphi \in {\bf CNF}(M)$ is determined regarding $S_i$
as specified in the previous paragraph. We show that $\psi_X$ is a
$m^{0.999}/3$-witnessing formula for $X$. Indeed, let $P$ be a path containing $X$.
Due to the uniformity, there is $S_i \in A(P)$ for some $1 \leq i \leq q$ (that is $A(P)$
satisfies $Conj(S_i)$.
Then, by definition of $\varphi_i$, $A(P)$ satisfies $\varphi_i$ 
and hence $A(P)$ satisfies $\psi_X$, concluding the proof. 

\begin{definition} [{\bf Partition of a path and a generating set}] \label{def:part}
Let $P$ be a path of $Z$.
A sequence $P_1, \dots, P_c$ of subpaths of $P$ is a \emph{partition} of $P$ if $P_1$ is a prefix of $P$,
for each $1<i \leq c$, the first vertex of $P_i$ is the last vertex of $P_{i-1}$, and $P_c$ is a suffix
of $P$.

We say that the set $X$ of ends of $P_1, \dots, P_{c-1}$ \emph{generates} $P_1,\dots, P_c$ on $P$.
\end{definition}

\begin{definition} [{\bf Variable separators}] \label{def:separ}
Let $X \subseteq V(Z)$
not containing the source nor the sink vertex of $Z$.

We say that $X$
\emph{separates} two disjoint subsets $Y_1$ and $Y_2$ of $Vars(Z)$ 
if there is a computational path $P$ of
$Z$ passing through all the vertices of $X$ such that one of the following
two statements is true regarding the partition $P_1, \dots, P_{|X|+1}$
generated by $X$ on $P$.
\begin{enumerate}
\item Elements of $Y_1$ occur only on paths $P_i$ with odd $i$ and elements of $Y_2$ 
occur only on paths $P_i$ with even $i$.
\item Elements of $Y_1$ occur only on paths $P_i$ with even $i$ and elements of $Y_2$ 
occur only on paths $P_i$ with odd $i$.
\end{enumerate}  
\end{definition}

The proofs of the following two lemmas is provided in the Appendix.

\begin{lemma} \label{lem:exsep}
On each computational path $P$ of $Z$ there is a set $X$ of vertices 
size at most $\log n/4000$ that separates two disjoint 
subsets of ${\bf Roots}(H)$ of size at least $m^{0.999}$.
\end{lemma}

\begin{lemma} \label{lem:allneg}
Let $Z$ be a $d$-{\sc nbp} and let $X$ be a subset of its vertices
separating two disjoint subsets $Y_1$ and $Y_2$ of $Vars(Z)$.
Let $Y'_1 \subseteq Y_1$ and $Y'_2 \subseteq Y_2$.
Let $S$ be an assignment to $Vars(Z) \setminus (Y_1 \cup Y_2)$
and $Q'$ and $Q''$ be two computational paths of $Z$ containing $X$ such
that $S \subseteq A(Q') \cap A(Q'')$, $Q'$ assigns negatively all the variables
of $Y'_1$ and $Q''$ assigns negatively all the variables of $Y'_2$.
Then there is a computational path $Q^*$ of $Z$ containing $X$
with $S \subseteq A(Q^*)$ that assigns negatively all the variables of $Y'_1 \cup Y'_2$.
\end{lemma}

\begin{lemma} \label{lem:fool}
Let $Z$ be a $d$-{\sc nbp} computing $\phi_H$.
Let $X \subseteq V(Z)$ be a subset of vertices of $Z$ separating
subsets $Y_1$ and $Y_2$ of ${\bf Roots}(H)$ and
let $M$ be a pseudomatching of $H$ between $Y_1$ and $Y_2$. 

Then for each set $S$ of literals over $V(H) \setminus \bigcup M$
there is $\varphi \in {\bf CNF}(M)$ such that for each path $P$
containing $X$ and with $S \subseteq A(P)$, $\varphi$ is satisfied by 
$A(P)$.
\end{lemma}

{\bf Proof.}
Fix a set $S$ of literals over $V(H) \setminus \bigcup M$.
For $\{t_i,t_j\} \in M$, let $C^{1/2}_{i \rightarrow j}$
be clauses whose sets of literals are, respectively,
$V(P^{1/2}_{i \rightarrow j})$ and $V(P^{1/2}_{j \rightarrow i})$. 

We first show that
for each $\{t_i,t_j\} \in M$, there is
$C'_{i,j} \in \{C^{1/2}_{i \rightarrow j},C^{1/2}_{j \rightarrow i})$
such that for each path $P$ containing $X$ with $S \subseteq A(P)$,
$C'_{i,j}$ is satisfied by $A(P)$.

Indeed, suppose this is not true.
This means that there are two computational paths $Q'$ and $Q''$
containing $X$ with $S \subseteq A(Q') \cap A(Q'')$ such that 
$Q'$ falsifies $C^{1/2}_{i \rightarrow j}$ and 
$Q''$ falsifies $C^{1/2}_{j \rightarrow i}$ (it is necessary that $Q' \neq Q''$
because otherwise $Q'$ falsifies both $C^{1/2}_{i \rightarrow j}$ and $C^{1/2}_{i \rightarrow j}$
and hence the whole $C_{i,j}$, a contradiction).
All the non-root variables of $C_{i,j}$ belong to 
$V(H) \setminus \bigcup M$ and hence occur in $S$.
As $Q'_1$ falsifies $C^{1/2}_{i \rightarrow j}$,
all the non-root variables of $C^{1/2}_{i \rightarrow j}$
occur negatively in $S$. Symmetrically,
all the non-root variables of $C^{1/2}_{j \rightarrow i}$
occur negatively in $S$. Therefore, all the non-root variables
of $C_{i,j}$ occur negatively in $S$.
Assume w.l.o.g. that $t_i \in Y_1$ and $t_j \in Y_2$.
Then, by Lemma \ref{lem:allneg}, there is a computational
path $Q^*$ of $Z$ passing through $X$ with $S \subseteq A(Q)$
and both $t_i$ and $t_j$ occurring negatively.
That is $Q^*$ falsifies $C_{i,j}$ in contradiction to the definition of $Z$
computing $\phi$. Thus $C'_{i,j}$, as specified in the previous paragraph indeed
exists. 

It follows that each computational path $P$ of $Z$ containing 
$X$ and with $S \subseteq A(P)$ satisfies the {\sc cnf} $\varphi$
whose clauses are $C'_{i,j}$, for every $\{t_i,t_j\} \in M$.
Clearly $\varphi \in {\bf CNF}(M)$.
$\blacksquare$

{\bf Proof of Theorem \ref{lbcons}}
Let $P$ be a computational path of $Z$. 
By Lemma \ref{lem:exsep}, $P$ contains a set $X$ of vertices of size at most
$\log n/4000$ that separates two disjoint subsets $Y_1$ and $Y_2$
of $Roots(H)$ of size at least $m^{0.999}$ each.

By the large pseudomatching property of pseudoexpanders,
there is a pseudomatching $M$ of $H$ between $Y_1$ and $Y_2$
with $|M| \geq m^{0.999}/3$.

Let $S_1, \dots, S_q$ be all sets of literals over $V(H) \setminus \bigcup M$.
By Lemma \ref{lem:fool}, for each $S_i$ there is $\varphi_i \in {\bf CNF}(M)$
that is satisfied by $A(Q)$ for each computational path $Q$ containing
$X$ provided that $S_i \subseteq A(Q)$.

Let $\psi=(Conj(S_1) \wedge \varphi_1) \vee \dots \vee (Conj(S_q) \wedge \varphi_q)$.
Clearly $\psi \in {\bf ORCNF}(M)$. It remains to show that $\psi$ is satisfied
by the assignment carried by each computational path $Q$ containing $X$.
For this notice that $A(Q)$ necessarily contains some $S_i$ (since $Z$ is a uniform $d$-{\sc nbp},
all the variables occur in $A(Q)$, hence this $S_i$ is just the projection of
$A(Q)$ to $V(H) \setminus \bigcup M$). Therefore, by the previous paragraph, $A(Q)$
satisfies $\varphi_i$ and hence also $Conj(S_i \wedge \varphi_i)$. 
Consequently, $A(Q)$ satisfies $\psi$.
$\blacksquare$

\section{Proof of Theorem \ref{lbsize}} \label{sec:lbsize}
Recall that Theorem \ref{lbsize} states that 
if $M_1, \dots, M_r$ are \emph{large} pseudomatchings
of $H \in {\bf PE}$ and $\psi_1, \dots, \psi_r$ are formulae
such that $\psi_i \in {\bf ORCNF}(M_i)$ for $1 \leq i \leq r$
and every satisfying assignment of $\phi_H$ satisfies one of 
$\psi_1, \dots, \psi_r$, then $r$ is exponentially large.
To prove this, we define a probability space over the set of satisfying
assignments of $\phi_H$. Then we fix a large matching $M$ and
$\psi \in {\bf ORCNF}(M)$ and show that the probability of satisfaction
of $\psi$ is exponentially small in $|M|$. By the union bound, 
Theorem \ref{lbsize} follows.

In order to prove the small probability of satisfaction of $\psi$
we introduce a special property of sets of literals over $V(H) \setminus M$
and show that: (i) the probability of the satisfaction of $\psi$ conditioned
on this property is exponentially small in $|M|$ and (ii) the probability that this
property is \emph{not} satisfied is exponentially small in $|M|$. 
Thus the probability of satisfaction of $\psi$ is upper bounded by the 
sum of these two exponentially small quantities and hence is exponentially small itself.
We provide further informal explanation of the use of the above special property
after it has been defined.

In this section we refer to the variables of $\phi_H$ that are 
root, leaf, and the rest of vertices of $V(H)$ as
the \emph{root}, \emph{leaf}, and \emph{internal} variables of $\phi_H$, respectively.
All sets of literals considered in this section are over subsets of 
variables of $\phi_H$. 

\begin{definition} [{\bf Fixed set of literals.}] \label{def:fixset}
A literal $\leaf_{i,j}$ (that is, the positive literal of the variable $\leaf_{i,j}$
which is the joint leaf of trees $T_i$ and $T_j$)
is \emph{fixed} w.r.t. a set $S$ of literals if the following two conditions are true.
\begin{itemize}
\item $\leaf_{i,j} \in S$;
\item the rest of variables of $C_{i,j}$ occur negatively in $S$.
\end{itemize}

We denote the set of fixed literals of $S$ by $Fix(S)$.
\end{definition}

\paragraph{Probability space for the set of satisfying assignments of $H$.}
Denote by ${\bf SAT}(H)$ the set of satisfying assignments of $\phi_H$,
which is the family of sets $S$ of literals with $Var(S)=V(H)$ that satisfy $\phi_H$.
In this section we use the probability space with the sample space ${\bf SAT}(H)$ and
for each $S \in {\bf SAT}(H)$, the probability of the elementary event $\{S\}$ is 
$(1/2)^{|S \setminus Fix(S)|}=(1/2)^{|V(H) \setminus Fix(S)|}$. In order to ensure that this definition is valid,
we prove that the probabilities of elementary events sum up to $1$.

\begin{proposition} [{\bf Validity of the probability space}] 
\label{prop:validsp}
$\sum_{S \in {\bf SAT}(H)} (1/2)^{|S \setminus Fix(S)|}=1$
\end{proposition}

The proof is postponed to the appendix.

The following is the main statement to be proved in this section.

\begin{theorem} \label{theor:mainprob}
There are constants $0<\lambda<1$ and $0<\mu<1$ such that for any sufficiently large $H \in {\bf PE}$
with $m=|{\bf Roots}(H)|$, any pseudomatching $M$ of $H$ of size at least $m^{0.999}/3$ and any
$\psi \in {\bf ORCNF}(M)$, the probability that a satisfying assignment of $\phi_H$ satisfies $\psi$ 
is at most $\lambda^{m^{\mu}}$.
\end{theorem}

Assuming that Theorem \ref{theor:mainprob} holds, we can now prove Theorem \ref{lbsize}.

{\bf Proof of Theorem \ref{lbsize}.}
Let $n=|V(H)|$. It is not hard to observe that $n \leq 3m^2$. Therefore, there is a constant $0<\gamma<1$
such that $\lambda^{m^{\mu}} \leq \lambda^{n^{\gamma}}$ and, hence, according to Theorem \ref{theor:mainprob},
the probability that a satisfying assignment of $\phi_H$ satisfies $\psi$ 
at most $\lambda^{n^{\gamma}}$. 

Let $\alpha=1/\lambda$ and assume by contradiction that $|{\bf \Psi}| < \alpha^{n^{\gamma}}$.
By the union bound, the probability that an assignment of $\phi_H$ satisfies one of the elements
of ${\bf \Psi}$ is at most $|{\bf \Psi}|*\lambda^{n^{\gamma}}<1$, the inequality follows from the assumption
that there is a satisfying assignment of $\phi_H$ that does not satisfy any element of ${\bf \Psi}$,
in contradiction to its definition.
$\blacksquare$

In the rest of this section we prove Theorem \ref{theor:mainprob}.

\paragraph{Events.}
The events are subsets of ${\bf SAT}(H)$.
\begin{itemize}
\item {\bf Containment event:} for a set $S$ of literals, ${\bf EC}(S)$ denotes the set of all elements 
$S' \in {\bf SAT}(H)$ such that $S \subseteq S'$.
\item {\bf Satisfiability event:} for a formula $\varphi$, ${\bf ES}(\varphi)$ denotes the set of all elements of
$S' \in {\bf SAT}(H)$ that satisfy $\varphi$. 
\end{itemize}

The following lemma will allow us to move back and forth
between the probability of the union of events and the sum of their probabilities.

\begin{lemma} \label{lem:disjcont}
Let $S_1$ and $S_2$ be two distinct sets of literals
such that $Var(S_1)=Var(S_2)$. Then 
${\bf EC}(S_1) \cap {\bf EC}(S_2)=\emptyset$ 
\end{lemma}

{\bf Proof.}
Suppose $S \in {\bf EC}(S_1) \cap {\bf EC}(S_2)$.
Then the projection of $S$ to $Var(S_1)$ is both
$S_1$ and $S_2$, which is impossible since $S_1 \neq S_2$.
$\blacksquare$

\begin{definition}[{\bf Siblings in $H$}]
Two internal vertices $u$ and $v$ of $H$ are \emph{siblings} if they belong to the same tree $T_i$ and
if $u$ and $v$ are siblings in $T_i$. 
\end{definition}
(This definition is unambiguous because each internal vertex belongs
to exactly one tree of ${\bf T}(H)$.

\begin{definition}[{\bf Set of literals respecting a pseudoedge}] \label{def:resp}
A set $S$ of literals \emph{respects} a pseudoedge $\{t_i,t_j\}$
if the following two conditions hold.
\begin{enumerate}
\item All non-root variables of $C_{i,j}$ occur negatively in $S$.
\item The siblings of all internal variables of $C_{i,j}$ occur positively
in $S$.
\end{enumerate}
\end{definition}

\begin{definition}[{\bf $\eta$-comfortable set of literals}] \label{def:etacomf}
For $0 \leq \eta \leq 1$, a set $S$ of literals is $\eta$-\emph{comfortable} 
w.r.t. a pseudomatching $M$ if the following two conditions hold.
\begin{enumerate}
\item $S$ does not falsify any clause of $\phi_H$
\item $S$ respects at least $\eta*|M|$ pseudoedges of $M$.
\end{enumerate}
\end{definition}

We continue the informal explanation started in the beginning of this section.
The special property on sets of assignments $S$ over $V(H) \setminus \bigcup M$
mentioned there is that $S$ is $\eta$-comfortable w.r.t. $M$ for $\eta=c*m^{-4/4.9}$
for a specially chosen constant $c$. 
We prove that (i) if $S$ is $\eta$-comfortable w.r.t. $M$ then for any
$\varphi \in {\bf CNF}(M)$, $Pr({\bf ES}(\varphi)|{\bf EC}(S)) \leq (2/3)^{\eta*|M|}$
and (ii) the probability that an element of ${\bf SAT}(H)$ is not $\eta$-comfortable
is exponentially small in $|M|$. 

The proof of (i) is, essentially, a reduction from the following: 
(iii) if $S$ is an assignment of $V(H) \setminus M$ that is $1$-comfortable
w.r.t. $M$ then $Pr({\bf ES}(\varphi)|{\bf EC}(S)) \leq (2/3)^{|M|}$.
To prove (iii), we note $S$ assigns all the variables of $\phi_H$ except
the root variables of $C_{i,j}$ for $\{t_i,t_j\} \in M$. Moreover, 
by Definition \ref{def:resp}, all the non-root variables of each such $C_{i,j}$
are assigned negatively. Therefore, for each such $C_{i,j}$, either $t_i$
or $t_j$ is assigned positively, thus making $3$ choices per clause and 
the total number of possible elements of ${\bf EC}(S)$ being at most $3^{|M|}$.
It may seem that some choices are not available because negative assignment
to some $t_i$, together with $S$, could falsify clauses corresponding to
pseudoedges that are not in $M$. However, this does not happen since 
positive assignments to siblings of internal variables of 
$C_{i,j}$ eliminate such a possibility, so the number of elements 
of ${\bf EC}(S)$ is indeed $3^{|M|}$.

Moreover, we show that each element of ${\bf EC}(S)$ has the same set of fixed
literals and has the same probability. Therefore, proving (iii)  amounts to
proving that at most $2^{|M|}$ elements of ${\bf EC}(S)$ satisfy $\varphi$.

In order to show this, 
recall that $\varphi \in {\bf CNF}(M)$ consists of `halves' of clauses
$C_{i,j}$ for $\{t_i,t_j\} \in M$. Hence every element of ${\bf EC}(S)$ satisfying
$\varphi$ must satisfy all the root variables of these halves, leaving $2$ choices
per clause (either positive or negative assignment of the other root variable)
and the total number of choices is at most $2^{|M|}$ as required.

To prove statement (ii) we define for each $\{t_i,t_j\} \in M$
a $1-0$ random variable $X_{i,j}$ which is $1$ exactly on those 
elements of ${\bf SAT}(H)$ that respect $\{t_i,t_j\}$.
We then prove that (iv) the probability that $X_{i,j}=1$ is at most $c_1*m^{-4/4.9}$
for some constant $c_1$ and that (v) the variables $X_{i,j}$ are mutually
independent. By Chernnoff's bound, with high probability, an element
of ${\bf SAT}(H)$ respects at least $c*m^{-4/4.9}*|M|$ edges of $M$ where
$c$ is a constant dependent on $c_1$ (and we choose $\eta=c*m^{-4/4.9}$), thus
implying (ii).

For statement (iv), note that because of the small height property
of pseudoexpanders, each clause $C_{i,j}$ of $\phi_H$ has about $2\log m/4.9$
variables. Together with about the same number of siblings of internal variables
of $C_{i,j}$, this means that assignments of about $4 \log m/4.9$ variables
of $\phi_H$ must be fixed in an element of ${\bf SAT}(H)$ in order to respect
$\{t_i,t_j\}$. Although, due to `non-uniform' character of the probability space, some
careful calculation is required, this is, essentially the reason why
the probability of respecting $t_{i,j}$ is $O(2^{-4 \log m/4.9})=m^{-4/4.9}$.

The underlying reason for statement (v) is that respecting of 
two distinct pseudoedges of a pseudomatching requires fixing assignments 
on \emph{disjoint} sets of variables of $\phi_H$ (though, again, in the 
considered probability space, this implication is not trivial and requires
a careful calculation). 

The statement (i) is Theorem \ref{theor:boundcomf}.
and it is proved in Subsection \ref{sec:comf}.
The statements (iv) and (v) are Theorem \ref{theor:indi}.
whose proof is provided in Subsection \ref{subsec:indi}.
We now show how Theorem \ref{theor:mainprob} follows
from (i) and (ii). 

The proof of Theorem \ref{theor:mainprob} amounts to showing that 
$\sum_{i=1}^q Pr({\bf ES}(\varphi_i) \cap {\bf EC}(S_i))$ is exponentially
small, where $S_1, \dots, S_q$ are all possible assignments over $V(H) \setminus \bigcup M$
and $\varphi_1, \dots, \varphi_q$ are elements of ${\bf CNF}(M)$ such that
$\psi=\bigvee_{i=1}^q (Conj(S_i) \wedge \varphi_i)$.
Because of the possibility of rearrangement, we can assume w.l.o.g. 
that for some $r \leq q$, $S_1, \dots, S_r$ are $\eta$-comfortable w.r.t. $M$
while $S_{r+1}, \dots, S_q$ are not. Then $\sum_{i=r+1}^q Pr({\bf ES}(\varphi_i) \cap {\bf EC}(S_i))$
is exponentially small by (ii) 
and  $\sum_{i=1}^r Pr({\bf ES}(\varphi_i) \cap {\bf EC}(S_i))=
      \sum_{i=1}^r Pr({\bf ES}(\varphi_i)|{\bf EC}(S_i))*Pr({\bf EC}(S_i))$.
All the conditional probabilities are exponentially small according to (i), therefore,
the last quantity can be upper bounded by a number exponentially small in $|M|$
multiplied by $\sum_{i=1}^r Pr({\bf EC}(S_i)) \leq \sum_{i=1}^q Pr({\bf EC}(S_i))=1$.
Thus as stated in the beginning of this section, $Pr({\bf ES}(\psi))$ can be upper bounded
by the sum of two exponentially small quantities.

\begin{lemma} \label{lem:projection}
A satisfying assignment $S$ of $\phi_H$ is $\eta$-comfortable
w.r.t. a pseudomatching $M$ of $H$ if and only if the projection $S'$ of $S$ to $V(H) \setminus \bigcup M$ is 
$\eta$-comfortable w.r.t. $M$.
\end{lemma}

{\bf Proof.}
Assume $S$ is $\eta$-comfortable w.r.t. $M$.
By definition, this means that $S$ contains particular occurrences of variables \emph{all of which are non-root ones}. Since $\bigcup M$ consists of root variables only,
these occurrences are preserved in $S'$. Therefore,
$S'$ is $\eta$-comfortable w.r.t. $M$.

Conversely, if $S'$ is $\eta$-comfortable w.r.t. $M$ then the witnessing occurrences of variables remain in any superset of $S'$, in particular in $S$.
As $S$ does not falsify any clause, we conclude that $S$ is $\eta$-comfortable w.r.t. $M$. $\blacksquare$

\begin{theorem} \label{theor:boundcomf}
Let $M$ be a pseudomatching of $H$, $0 \leq \eta \leq 1$.
Then for any set $S$ of literals over $(V(H) \setminus \bigcup M)$, which is $\eta$-comfortable w.r.t. to 
$M$ and for any $\varphi \in {\bf CNF}(M)$,
$Pr({\bf ES}(\varphi)|{\bf EC}(S)) \leq (2/3)^{\eta*|M|}$.
\end{theorem}

The proof is provided in Subsection \ref{sec:comf}

\paragraph{Indicator variables for pseudoedges and their sums.}
Let $\{t_i,t_j\}$ be a pseudoedge of $H$. The indicator variable $X_{i,j}$ is a $1-0$ variable
such that for $S \in {\bf SAT}(H)$, $X_{i,j}(S)=1$ if and only if $S$ respects $\{t_i,t_j\}$.

For a pseudomatching $M$ of $H$, let
$X_M=\sum_{\{t_i,t_j\} \in M} X_{i,j}$.

\begin{theorem}[{\bf Statements about indicator variables.}]
\label{theor:indi}
\begin{enumerate}
\item For any pseudoedge $\{t_i,t_j\}$ of $H$, $Pr(X_{i,j}=1) \geq 3/128*m^{-4/4.9}$.
\item For any pseudomatching $M$, the variables $\{X_{i,j}|\{t_i,t_j\} \in M\}$ are mutually
independent. 
\end{enumerate}
\end{theorem}

The proof appears in Subsection~\ref{subsec:indi}.

\begin{lemma} \label{lem:chern}
Let $M$ be a pseudomatching of $H$ and $\eta=3/12800*m^{-4/4.9}$.
Then $Pr(X_M<\eta*|M|) \leq 0.4^{100*\eta*|M|}$.
\end{lemma}

{\bf Proof.}
By Theorem \ref{theor:indi}, the variables
$\{X_{i,j}|\{t_i,t_j\} \in M\}$ are mutually independent, which allows us to use Chernoff's bounds.

For this, we first calculate a lower bound on the expected value of $X_M$. 

\begin{equation} \label{eq:exp}
E[X_M]=\sum_{\{t_i,t_j\} \in M} E[X_{i,j}]=
\sum_{\{t_i,t_j\} \in M} Pr(X_{i,j}=1) \geq 100*\eta*|M|
\end{equation}

where the last inequality follows from Theorem \ref{theor:indi}
stating that $Pr(X_{i,j}=1) \geq 100*\eta$. 

To apply Chernoff's bounds,
we use inequality (4.4.) of Theorem 4.5. of \cite{MUbook}. We provide it
here for convenience, adapted to our notation

\begin{equation} \label{eq:chern}
Pr(X_M \leq (1-\delta)*E[X_M]) \leq \left(\frac{e^{-\delta}}{(1-\delta)^{1-\delta}}\right)^{E[X_M]}
\end{equation}
for any $0<\delta<1$.
\sloppy
Now, we obtain the following.
\begin{multline}
Pr(X_M<\eta*|M|) \leq Pr(X_M \leq E[X_M]/100)=\\
Pr(X_M \leq (1-99/100)*E[X_M]) \leq 0.4^{E[X_M]} \leq 0.4^{100*\eta*|M|}
\end{multline}
where the first inequality follows from \eqref{eq:exp}
and the obvious fact that 
$Pr(X_M< E[X_M]/100) \leq Pr(X_M \leq E[X_M]/100)$, 
the second inequality follows
from substituting $0.99$ to $\delta$ in \eqref{eq:chern} and  the inequality
$\frac{e^{-0.99}}{0.01^{0.01}} \leq 0.4$ verified by a straightforward calculation, the third inequality follows from \eqref{eq:exp}.
$\blacksquare$

{\bf Proof of Theorem \ref{theor:mainprob}.}
By definition $\psi=Conj(S_1) \wedge \varphi_1 \vee \dots \vee Conj(S_q) \wedge \varphi_q$, where
$S_1, \dots, S_q$ are sets of literals over $(V(H) \setminus \bigcup M)$ 
and $\varphi_i \in {\bf CNF}(M)$
for $1 \leq i \leq q$.

Observe that
\begin{equation} \label{eq:satpsi}
{\bf ES}(\psi)=\bigcup_{i=1}^q [{\bf ES}(\varphi_i) \cap {\bf EC}(S_i)]
\end{equation}

Indeed, assume $S \in {\bf ES}(\psi)$. Then $S$ satisfies one of the conjuncts of
$\psi$. Hence, there is $1 \leq i \leq q$ such that
$S \in {\bf ES}(Conj(S_i) \wedge \varphi_i)$. 
Therefore, $S$ satisfies \emph{both}
$Conj(S_i)$ and $\varphi_i$.
Thus, $S \in {\bf ES}(Conj(S_i)) \cap {\bf ES}(\varphi_i)$.
To satisfy a conjunction of literals, $S$ must contain all these literals, as a result
$S \in {\bf EC}(S_i)$. Consequently, $S \in  {\bf ES}(\varphi_i) \cap {\bf EC}(S_i)$
and hence $S \in \bigcup_{i=1}^q [{\bf ES}(\varphi_i) \cap {\bf EC}(S_i)]$.

Conversely, assume that 
$S \in \bigcup_{i=1}^q [{\bf ES}(\varphi_i) \cap {\bf EC}(S_i)]$.
Then $S \in {\bf ES}(\varphi_i) \cap {\bf EC}(S_i)$ for some $1 \leq i \leq q$.
Since $S$ contains all the literals of $S_i$, $S$ satisfies the conjunction of
these literals. Therefore, $S \in {\bf ES}(Conj(S_i))$. Since $S$ satisfies
both $\varphi_i$ and $Conj(S_i)$, $S$ also satisfies their conjunction.
Consequently, $S \in {\bf ES}(Conj(S_i) \wedge \varphi_i)$. Finally, since 
$Conj(S_i) \wedge \varphi_i$ is a disjunct of $\psi$, we conclude that $S \in {\bf ES}(\psi)$.

Set $\eta=3/12800*m^{-4/4.9}$. Assume w.l.o.g., because of the possibility of rearrangement,
that $S_1, \dots, S_r$ are the $\eta$-comfortable w.r.t. $M$ 
and $S_{r+1}, \dots, S_q$ are not $\eta$-comfortable.

Then
\begin{multline} \label{eq:firstred}
Pr({\bf EC}(\psi))=Pr(\bigcup_{i=1}^q [{\bf ES}(\varphi_i) \cap {\bf EC}(S_i)])=\sum_{i=1}^q Pr({\bf ES}(\varphi_i) \cap {\bf EC}(S_i))=\\
\sum_{i=1}^r Pr({\bf ES}(\varphi_i) \cap {\bf EC}(S_i))+\sum_{i=r+1}^q Pr({\bf ES}(\varphi_i) \cap {\bf EC}(S_i))
\end{multline}

where the first equality follows from \eqref{eq:satpsi}. For the second equality, notice that
by Lemma \ref{lem:disjcont}, ${\bf EC}(S_i)$ for $1 \leq i \leq q$ are mutually disjoint.
Thereofre, ${\bf ES}(\varphi_i) \cap {\bf EC}(S_i))$ for $1 \leq i \leq q$ are also mutually
disjoint and hence the probability of their union can be replaced by the sum of their probabilities.
The third equality is correct because the right-hand part of its is a regrouping of items in
the left-hand part.

We are now, going to show that both $\sum_{i=1}^r Pr({\bf ES}(\varphi_i) \cap {\bf EC}(S_i))$
and $\sum_{i=r+1}^q Pr({\bf ES}(\varphi_i) \cap {\bf EC}(S_i))$ are exponentially small from where
the theorem will immediately follow. 

For the former, 
\begin{multline} \label{eq:comf}
\sum_{i=1}^r Pr({\bf ES}(\varphi_i) \cap {\bf EC}(S_i))=\sum_{i=1}^r [Pr({\bf ES}(\varphi_i)|{\bf EC}(S_i))*Pr({\bf EC}(S_i)] \leq \\
\sum_{i=1}^r [(2/3)^{\eta*|M|}*Pr({\bf EC}(S_i)]=(2/3)^{\eta*|M|}*Pr(\bigcup_{i=1}^r {\bf EC}(S_i)) \leq (2/3)^{\eta*|M|}
\end{multline}

where the first equality follows from the definition of conditional probability and the first 
inequality follows from Theorem \ref{theor:boundcomf}. The second equality is the result of moving
$(2/3)^{\eta*|M|}$ out of the brackets and replacing the sum of probabilities by the probability of union,
the replacement enabled by Lemma \ref{lem:disjcont} (see the explanation to \eqref{eq:firstred}
for the detailed justification). Finally, the second inequality holds because 
$Pr(\bigcup_{i=1}^r {\bf EC}(S_i)) \leq 1$.

In order to establish an exponentially small lower bound on 
$\sum_{i=r+1}^q Pr({\bf ES}(\varphi_i) \cap {\bf EC}(S_i))$, the key observation is the following.

\begin{equation} \label{eq:leadchern}
\bigcup_{i=r+1}^q {\bf EC}(S_i)=\{S'| (S' \in {\bf SAT}(H)) \wedge (X_M(S') < \eta*|M|)\}
\end{equation}

Indeed, assume that $S \in \bigcup_{i=r+1}^q {\bf EC}(S_i)$.
Then, $S_i \subseteq S$ for some $r+1 \leq i \leq q$.
By our assumption, $S_i$ is not $\eta$-comfortable w.r.t. $M$.
Hence, as $S$ is a satisfying assignment of $\phi_H$ 
$S$ is not $\eta$-comfortable w.r.t. $M$ according 
to Lemma \ref{lem:projection}. 

As $S$ is a satisfying assignment of $\phi_H$, 
the only way for $S$ to be not $\eta$-comfortable is to
respect less than $\eta*|M|$ pseudoedges $\{t_i,t_j\}$ of $M$.
Consequently, less than $\eta*|M|$ corresponding 
variables $X_{i,j}$ equal $1$
on $S$ and hence $X_M(S) <\eta*|M|$. That is, $S \in 
\{S'| (S' \in {\bf SAT}(H)) \wedge (X_M(S') < \eta*|M|)\}$.

Conversely, assume that 
$S \in \{S'| (S' \in {\bf SAT}(H)) \wedge (X_M(S') < \eta*|M|)\}$
That is, $X_M(S)<\eta*|M|$. Hence less than $\eta*|M|$ variables
$\{X_{i,j}|\{t_i,t_j\} \in M\}$ are $1$ on $S$, that is, in turn,
$S$ respects less than $\eta*|M|$ pseudoedges of $M$, hence
$S$ is not $\eta$-comfortable w.r.t. $M$. As $S$ is a satisfying assignment
of $\phi_H$, the projection $S^*$ of $S$ to $V(H) \setminus \bigcup M$
is not $\eta$-comfortable by Lemma \ref{lem:projection}.
By assumption, there is $r+1 \leq i \leq q$ such that
$S_i=S^*$. Hence, $S \in {\bf EC}(S_i)$ as required.

Now, we obtain the following:
\begin{multline} \label{eq:finchern}
\sum_{i=r+1}^q Pr({\bf ES}(\varphi_i) \cap {\bf EC}(S_i)) \leq \sum_{i=r+1}^q Pr({\bf EC}(S_i))=\\
Pr(\bigcup_{i=r+1}^q {\bf EC}(S_i))=Pr(\{S'| (S' \in {\bf SAT}(H)) \wedge (X_M(S') < \eta*|M|)\})=\\
Pr(X_M<\eta*|M|) \leq 0.4^{100*\eta*|M|}
\end{multline}

For the first inequality, note that for each $r+1 \leq i \leq q$, 
${\bf ES}(\varphi_i) \cap {\bf EC}(S_i) \subseteq {\bf EC}(S_i)$, hence the probability of the 
event on the left-hand side does not exceed the probability of the event on the right-hand side.
The first equality follows from Lemma \ref{lem:disjcont} (see the explanation to \eqref{eq:firstred}
for the detailed justification). The second equality is just \eqref{eq:leadchern}.
The third equality is just an effect of changing notation for the probability of the same event. 
The fourth inequality follows from Lemma \ref{lem:chern}.

By substituing \eqref{eq:comf} and \eqref{eq:finchern} into the last item of \eqref{eq:firstred}
and then, by substituting the value of $\eta$ and the assumed lower bound of $|M|$, we obtain
the following.
\begin{multline}
Pr({\bf EC}(\psi)) \leq (2/3)^{\eta*|M|}+0.4^{100*\eta*|M|} \leq\\ 
(2/3)^{3/12800*m^{-4/4.9}*m^{0.999}/3}+(0.4)^{100*3/12800*m^{-4/4.9}*m^{0.999}/3}
\end{multline}

Let $\mu=0.999-4/4.9$. Clearly, $0<\mu<1$ and there is $\lambda$ such that, for a sufficiently large $m$,
$Pr({\bf EC}(\psi)) \leq \lambda^{m^{\mu}}$. 
$\blacksquare$

\subsection{Proof of Theorem \ref{theor:boundcomf}} \label{sec:comf}
Referring to the informal explanation provided after Definition \ref{def:etacomf},
Theorem \ref{theor:boundcomf} is statement (i). The proof strategy outlined there
is implemented in this section as follows. 
 
Using the next lemma we will be able to conclude that
if $S$ is a partial assignment on $V(H) \setminus \bigcup M$
that is $1$-comfortable for $M$, then all $3^{|M|}$ elements of ${\bf EC}(S)$
are satisfying assignments of $\phi_H$. Next, Lemma \ref{lem:sameprob} states that all
these elements have the same probability. After that, in Lemma \ref{lem:almcomf},
we establish the claim of Theorem \ref{theor:boundcomf} for $\eta=1$ and then,
in the actual proof of Theorem \ref{theor:boundcomf}, we do `reduction' to an 
arbitrary $\eta$. 

\begin{lemma} \label{lem:allpos}
Let $M$ be a pseudomatching and let $S$ be a partial assignment on $V(H) \setminus \bigcup M$
that is $1$-comfortable for $M$ (that is, $S$ does not falsify any clause of $\phi_H$
and respects all the pseudoedges of $M$).
Let $\{t_i,t_j\} \in M$ and let $\{t_i,t_{j'}\} \neq \{t_i,t_j\}$ be another pseudoedge of 
$H$ (not contained in $M$ due to having a joint end with $\{t_i,t_j\}$). 
Then $S$ assigns positively at least one internal variable of $C_{i,j'}$.
\end{lemma}

{\bf Proof.}
The root $t_i$ of $T_i$ is a common ancestor of $\leaf_{i,j}$ and of $\leaf_{i,j'}$.
Therefore, we can identify the \emph{lowest} common ancestor $u$ of $\leaf_{i,j}$
and $\leaf_{i,j'}$ in $T_i$. Then $u$ has two children $v$ and $w$ (if it had one child
that child would be a common ancestor of $\leaf_{i,j}$ and $\leaf_{i,j'}$ in contradiction
to being $u$ the lowest one. Assume w.l.o.g. that $v$ is an ancestor of $l_{i,j}$.
Then $v$ belongs to $P_{i,j}$, hence $v$ is a variable of $C_{i,j}$. As $v$ is not a root
variable (due to having a parent), $v$ occurs negatively in $S$. Notice that $v \neq \leaf_{i,j}$
because $v$ has a sibling in $T_i$ while $\leaf_{i,j}$ does not. Therefore, $v$ is an internal
variable of $C_{i,j}$ and hence its sibling $w$ occurs positively in $S$.

Observe that $w$ is an ancestor of $\leaf_{i,j'}$. Indeed, otherwise, $v$ is an ancestor 
of $\leaf_{i,j'}$ in contradiction to $u$ being the lowest common ancestor of $\leaf_{i,j}$
and $\leaf_{i,j'}$. This means that $w$ belongs to $P_{i,j'}$ and hence $w$ is a variable
of $C_{i,j'}$. Furthermore, $w \neq \leaf_{i,j'}$ because $w$ has a sibling $v$ while $\leaf_{i,j'}$
does not have siblings. As $w \neq t_i$ due to having a parent, we conclude that $w$ is an internal
variable of $C_{i,j'}$ assigned positively by $S$, as required.
$\blacksquare$

\begin{lemma} \label{lem:sameprob}
Let $M$ be a pseudomatching and let $S$ be a partial assignment on $V(H) \setminus \bigcup M$
that is $1$-comfortable for $M$. Then for every $S' \in {\bf EC}(S)$, $Pr(\{S'\})=1/2^{|V(H) \setminus Fix(S)|}$
\end{lemma}

{\bf Proof.}
By definition, $Pr(\{S'\})=1/2^{|V(H) \setminus Fix(S')|}$. We are going to show that $Fix(S)=Fix(S')$.

Assume that $\leaf_{i,j} \in Fix(S)$. This means that $S$ occurs positively in $S$ and the rest of variables of $C_{i,j}$
occur negatively in $S$. As $S \subseteq S'$, all these occurrences are preserved in $S'$ and hence 
$\leaf_{i,j} \in Fix(S')$. That is $Fix(S) \subseteq Fix(S')$. 

Conversely, assume that $\leaf_{i,j} \in Fix(S')$. Then $\{t_i,t_j\} \notin M$. Indeed, otherwise, as $S$ respects $\{t_i,t_j\}$,
$\leaf_{i,j}$ is assigned negatively in $S$ and hence in $S'$ in contradiction to the assumption that $\leaf_{i,j} \in Fix(S')$.

Furthermore, it cannot happen that $|\{t_i,t_j\} \cap \bigcup M|=1$. Indeed, otherwise, there is a pseudoedge, say
$\{t_i,t_j'\} \in M$ such that $\{t_i,t_j\} \neq \{t_i,t_j'\}$.
Then, by Lemma \ref{lem:allpos}, at least one non-leaf variable of $C_{i,j}$ is assigned positively by $S$ and hence by $S'$.
However, this is a contradiction to $\leaf_{i,j} \in Fix(S')$ requiring all the non-leaf variables of $C_{i,j}$ to be assigned
negatively. 

It remains to assume that $\{t_i,t_j\} \cap \bigcup M=\emptyset$. But in this case all the variables of $C_{i,j}$ are
contained in $V(H) \setminus \bigcup M$ and hence occur in $S$. In particular, $\leaf_{i,j}$ 
occurs positively in $S$ and the rest of variables of $C_{i,j}$
occur negatively in $S$. It follows that $\leaf_{i,j}\in Fix(S)$. That is $Fix(S') \subseteq Fix(S)$.
$\blacksquare$

\begin{lemma} \label{lem:almcomf}
Let $M$ be a pseudomatching and let $S$ be a partial assignment on $V(H) \setminus \bigcup M$
that is $1$-comfortable for $M$. Let $\varphi \in {\bf CNF}(M)$. 
Then $Pr({\bf ES}(\varphi)|{\bf EC}(S))=(2/3)^{|M|}$.
\end{lemma}

{\bf Proof.}
Let ${\bf SA}$ be the set of all extensions of $S$ to $V(H)$ such that for each $\{t_i,t_j\} \in M$,
at least one of $t_i,t_j$ occurs positively in $S$.
Observe that
\begin{equation} \label{eq:sa}
{\bf EC}(S)={\bf SA}
\end{equation}

Indeed, let $S' \in {\bf EC}(S)$. Then $S'$ is a satisfying assignment of $\phi_H$. In particular,
for each $\{t_i,t_j\} \in M$, the clause $C_{i,j}$ is satisfied. However, all the variables
of $C_{i,j}$, except the root ones occur negatively in $S$ and hence in $S'$. It follows that, to satisfy
$C_{i,j}$, at least one of root variables $t_i,$ or $t_j$ must occur positively in $S'$.
Consequently (taking into account that $S'$ is an extension of $S$), 
$S' \in {\bf SA}$ and hence ${\bf EC}(S) \subseteq {\bf SA}$. 

Conversely, assume that $S' \in {\bf SA}$. To show that $S' \in {\bf EC}(S)$, we need to demonstrate
that $S'$ is a satisfying assignment of $\phi_H$.

For each $\{t_i,t_j\} \in M$, the clause $C_{i,j}$ is satisfied because, by definition of ${\bf SA}$,
at least one of $t_i,t_j$ occurs positively in $S'$. Consider a clause $C_{i,j'}$ such that 
$i \in \bigcup M$ while $j' \notin \bigcup M$. By Lemma \ref{lem:allpos}, at least one variable of $C_{i,j'}$
occurs positively in $S$ and hence in $S'$, satisfying $C_{i,j'}$. Finally consider a clause
$C_{i,j}$ such that none of $t_i,t_j$ occur in $\bigcup M$. This means that all the variables of $C_{i,j}$
occur in $S$. As $S$ is $1$-comfortable for $M$, $S$ does not falsify $C_{i,j}$ and, since $S$ contains literals
of all the variables of $C_{i,j}$, $S$ must \emph{satisfy} $C_{i,j}$ and hence so is $S'$. 
Consequently, $S'\in {\bf EC}(S)$ and hence ${\bf SA} \subseteq {\bf EC}(S)$.

As ${\bf SA}$ allows three independent choices per clause corresponding to a pseudoedge of $M$, we conclude that:
\begin{equation} \label{eq:sasize}
|{\bf SA}|=3^{|M|}
\end{equation}
Now, we are ready to obtain an expression for the probability of ${\bf EC}(S)$.

\begin{multline} \label{eq:probsa}
Pr({\bf EC}(S))=\sum_{S' \in {\bf EC}(S)} Pr(\{S'\})=\sum_{S' \in {\bf EC}(S)}(1/2)^{|V(H) \setminus Fix(S)|}=\\
(1/2)^{|V(H) \setminus Fix(S)|}*|{\bf EC}(S)|=(1/2)^{|V(H) \setminus Fix(S)|}*|{\bf SA}|=\\
(1/2)^{|V(H) \setminus Fix(S)|}*3^{|M|}
\end{multline}

where the second equality follows from definition of the probability of an event, the second equality
follows from Lemma \ref{lem:sameprob}, the third equality a straightforward algebraic transformation,
the fourth equality follows from \eqref{eq:sa}, the fifth equality follows from \eqref{eq:sasize}. 

Recall, that by definition of ${\bf CNF}(M)$, the clauses $C^{1/2}_{i,j}$ of $\varphi$ correspond to the pseudoedges
$\{t_i,t_j\} \in M$. In particular, the clause $C^{1/2}_{i,j}$ 
is a subclause of $C_{i,j}$ whose variables are either $V(P^{1/2}_{i \leftarrow j})$ or
$V(P^{1/2}_{j \leftarrow i})$. That is, either the vertices of the path from $t_i$ to $\leaf_{i,j}$ in $T_i$
or of the path from $t_j$ to $\leaf_{i,j}$ in $T_j$. By construction, $C^{1/2}_{i,j}$ contains exactly one
vertex of $\{t_i,t_j\}$. Denote this vertex by $root(i,j)$.

Let ${\bf SB}$ be the set of all extensions $S'$ of $S$ to $V(H)$ such that for each $\{t_i,t_j\} \in M$,
$root(i,j)$ occurs positively in $S'$.

\begin{equation} \label{eq:sb}
{\bf EC}(S) \cap {\bf ES}(\varphi)={\bf SB}
\end{equation}

Indeed, let $S' \in {\bf EC}(S) \cap {\bf ES}(\varphi)$.
This means that $S'$ is an extension of $S$ and
satisfies $C^{1/2}_{i,j}$ for all $\{t_i,t_j\} \in M$.
For each $C^{1/2}_{i,j}$ all the non-root vairables occur negatively
in $S$ and hence so in $S'$. That is, the only way for $S'$ to satisfy
$C^{1/2}$ is to contain the positively occurrence of $root(i,j)$.
Hence, $S' \in {\bf SB}$, and therefore ${\bf EC}(S) \cap {\bf ES}(\varphi) \subseteq {\bf SB}$.
Conversely, suppose $S' \in {\bf SB}$. Then $S'$ is an extension of $S$ to $V(H)$ in which
for each $\{t_i,t_j\}$, at least one variable occurs positively, namely $root(i,j)$.
This means that $S' \in {\bf SA}$ and hence $S' \in {\bf EC}(S)$ according to \eqref{eq:sa}.
Moreover, because of the positive occurrence of $root(i,j)$, for each $\{t_i,t_j\} \in M$,
$S'$ satisfies all the clauses $C^{1/2}_{i,j}$ and hence $S'$ satisfies $\varphi$. Consequently,
$S' \in {\bf ES}(\varphi)$. That is, $S' \in {\bf EC}(S) \cap {\bf ES}(\varphi)$ and hence
${\bf SB} \subseteq {\bf EC}(S) \cap {\bf ES}(\varphi)$. The correctness proof for \eqref{eq:sb}
is now complete. 

Since an element of ${\bf SB}$ can be formed by two independent choices
per clause corresponding to an element of $M$, we conclude that
\begin{equation} \label{eq:sbsize}
|{\bf SB}|=2^{|M|} 
\end{equation}

\begin{multline} \label{eq:probsb}
Pr({\bf EC}(S) \cap {\bf ES}(\varphi))=\sum_{S' \in {\bf EC}(S) \cap {\bf ES}(\varphi)} Pr(\{S'\})=\sum_{S' \in {\bf EC}(S) \cap {\bf ES}(\varphi)}(1/2)^{|V(H) \setminus Fix(S)|}=\\
(1/2)^{|V(H) \setminus Fix(S)|}*|{\bf EC}(S) \cap {\bf ES}(\varphi)|=(1/2)^{|V(H) \setminus Fix(S)|}*|{\bf SB}|=\\
(1/2)^{|V(H) \setminus Fix(S)|}*2^{|M|}
\end{multline}

where the reasoning is analogous to \eqref{eq:probsa} with \eqref{eq:sb} is used instead \eqref{eq:sa} and \eqref{eq:sbsize}
is used instead of \eqref{eq:sasize}.

Now, with \eqref{eq:probsa} and \eqref{eq:probsb} in mind we obtain the following.
\begin{multline}
Pr({\bf ES}(\varphi)|{\bf EC}(S))=Pr({\bf EC}(S) \cap {\bf ES}(\varphi))/Pr({\bf EC}(S))=\\
((1/2)^{|V(H) \setminus Fix(S)|}*2^{|M|})/((1/2)^{|V(H) \setminus Fix(S)|}*3^{|M|})=(2/3)^{|M|}
\end{multline}
$\blacksquare$

{\bf Proof of Theorem \ref{theor:boundcomf}.}
Let $M'$ be a subset of $M$ consisting of at least $\eta*|M|$ pseudoedges respected by $S$.
Let $\varphi'$ be the sub-CNF of $\varphi$ consisting of the clauses $C^{1/2}_{i,j}$ of $\varphi$
corresponding to all the pseudoedges $\{t_i,t_j\} \in M'$.
Clearly $\varphi' \in {\bf CNF}(M')$. 

Let ${\bf S}$ be the set of all extensions of $S$ to $V(H) \setminus \bigcup M'$ that do not falsify 
any clause. Observe that
\begin{equation} \label{eq:simplext}
{\bf EC}(S)=\bigcup_{S^* \in {\bf S}} {\bf EC}(S^*)
\end{equation}

Indeed, assume that $S' \in {\bf EC}(S)$.
Then the projection $S^*$ of $S'$ to $V(H) \setminus \bigcup M'$ is an extension of 
$S$ to $V(H) \setminus \bigcup M'$ that does not falsify any clauses of $\phi_H$, because
$S'$ does not falsify any. 
Therefore, $S^* \in {\bf S}$. Clearly, $S' \in {\bf EC}(S^*) \subseteq 
\bigcup_{S^* \in {\bf S}} {\bf EC}(S^*)$

For the opposite direction, let $S' \in \bigcup_{S^* \in {\bf S}} {\bf EC}(S^*)$.
Then there is $S^* \in {\bf S}$ such that $S' \in {\bf EC}(S^*)$. As $S \subseteq S^*$,
clearly, ${\bf EC}(S^*) \subseteq {\bf ES}(S)$. Consequently, $S \in {\bf EC}(S)$.

Observe that each $S^* \in {\bf S}$ respects each pseudoedge $\{t_i,t_j\}$ of $M'$.
Indeed, $S$ respects $\{t_i,t_j\}$ by assumption and, since $S \subseteq S^*$,
all literal occurrences witnessing that $S$ respects $\{t_i,t_j\}$ are retained in $S^*$. Since $S^*$ does not falsify 
any clause of $\phi_H$, $S^*$ is $1$-comfortable w.r.t. $M'$. Consequently, by Lemma \ref{lem:almcomf},

\begin{equation} \label{eq:probext}
\forall S^* \in {\bf S}, Pr({\bf ES}(\varphi')|{\bf EC}(S^*))=(2/3)^{|M'|}
\end{equation}

Now we obtain the following.
\begin{multline} \label{eq:finext}
Pr({\bf ES}(\varphi) \cap {\bf EC}(S)) \leq Pr({\bf ES}(\varphi') \cap {\bf EC}(S))=Pr({\bf ES}(\varphi') \cap \bigcup_{S^* \in {\bf S}} {\bf EC}(S^*))=\\
Pr(\bigcup_{S^* \in {\bf S}} ({\bf ES}(\varphi') \cap {\bf EC}(S^*)))=\sum_{S^* \in {\bf S}} Pr({\bf ES}(\varphi') \cap {\bf EC}(S^*))=\\
\sum_{S^* \in {\bf S}} (Pr({\bf ES}(\varphi')|{\bf EC}(S^*))*Pr({\bf EC}(S^*))=\sum_{S^* \in {\bf S}} ((2/3)^{|M'|}*Pr({\bf EC}(S^*))=
(2/3)^{|M'|}*Pr(\bigcup_{S^* \in {\bf S}} {\bf EC}(S^*))=\\(2/3)^{|M'|}*Pr({\bf EC}(S)) \leq (2/3)^{\eta*|M|}*Pr({\bf EC}(S))
\end{multline}

where the first inequality is correct because, due to $\varphi'$ being a sub-CNF of $\varphi$, any satisfying assignment of $\varphi$
is also a satisfying assignment of $\varphi'$. The first equality follows from \eqref{eq:simplext}. 
The second equality is a standard set-theoretic transformation.
The third equality follows from Lemma \ref{lem:disjcont} applied to the elements of ${\bf S}$. 
The fourth equality follows from definition of conditional probability. The fifth equality follows from Lemma \ref{lem:almcomf}.
The sixth equality is a result of moving $(2/3)^{|M'|}$ outside the brackets and the replacement of $\sum_{S^* \in {\bf S}} Pr({\bf EC}(S^*))$
by $Pr(\bigcup_{S^* \in {\bf S}} {\bf EC}(S^*))$ allowed by Lemma \ref{lem:disjcont}. Finally, the second inequality follows from definition 
of $|M'|$.

Thus \eqref{eq:finext} derives that $Pr({\bf ES}(\varphi) \cap {\bf EC}(S)) \leq (2/3)^{\eta*|M|}*Pr({\bf EC}(S))$.
The theorem immediately follows from dividing both parts of this inequality by $Pr({\bf EC}(S))$.
$\blacksquare$

\subsection{Proof of Theorem \ref{theor:indi}}\label{subsec:indi}

Recall that Theorem \ref{theor:indi} consists of two statements.
The first one states that $Pr(X_{i,j}) \geq 3/128*m^{-4/4.9}$, 
for every pseudoedge $\{t_i,t_j\} \in H$. The second statement claims
that, for every pseudomatching $M$ of $H$, the variables $X_{i,j}$ for
pseudoedges $\{t_i,t_j\} \in M$ are mutually independent.

For a pseudoedge $\{t_i,t_j\}$, denote by $S_{i,j}$ the set of literals
consisting of negative literals of all the non-root variables of 
$C_{i,j}$ and positive literals for all the siblings of the internal variables
of $C_{i,j}$.

It follows from the definition of $X_{i,j}$ that
\begin{observation} \label{obs:ij}
$Pr(X_{i,j}=1)=Pr({\bf EC}(S_{i,j})$.
\end{observation}

\begin{definition} [{\bf Guarded set of literals}] \label{def:guard}
A set $S$ of literals is \emph{guarded} from a clause $C_{i,j}$ of $\phi_H$
if one of the following three conditions holds.
\begin{enumerate}
\item $\leaf_{i,j}$ does not occur in $S$.
\item A non-leaf variable of $C_{i,j}$ occurs positively in $S$.
\item $\leaf_{i,j}$ occurs positively in $S$ and the rest of variables of $C_{i,j}$
occur negatively in $S$.
\end{enumerate}
\end{definition}

\begin{lemma} \label{lem:guardprob}
Let $S$ be a guarded set if literals.
Then $Pr({\bf EC}(S))=2^{-|S \setminus Fix(S)|}$.
\end{lemma}

The proof of Lemma \ref{lem:guardprob} appears in Section~\ref{sec:guard}.

We show that $Pr({\bf EC}(S_{i,j})$ can be represented as the sum of probabilities of ${\bf EC}$-events
for guarded
set of literals. Denote the set $\{S_{i,j} \cup \{t_i,t_j\}, S_{i,j} \cup 
\{\neg t_i,t_j\}, S_{i,j} \cup \{t_i,\neg t_j\}\}$
by ${\bf POSX}_{i,j}$. 

\begin{lemma} \label{lem:pos}
\begin{enumerate}
\item Each $S \in {\bf POSX}_{i,j}$ is guarded with $Fix(S)=\emptyset$.
\item $Pr(X_{i,j}=1)=Pr(\cup_{S \in {\bf POSX}_{i,j}} {\bf EC}(S))=
\sum_{S \in {\bf POSX}_{i,j}} Pr({\bf EC}(S))$.
\end{enumerate}
\end{lemma}

{\bf Proof.}
$S$ is guarded from $C_{i,j}$ because
at least one of $t_i,t_j$ occurs positively and hence the second condition
of Definition \ref{def:guard} is satisfied. No leaf variable except $\leaf_{i,j}$
occurs in $S$ hence, $S$ is guarded from the rest of the clauses according to the
first condition of Definition \ref{def:guard}. As $S$ does not contain positive literals
of leaf variables, it does not contain fixed literals.

The second equality of the second statement immediately follows
from Lemma \ref{lem:disjcont}. In light of Observation \ref{obs:ij},
for the first equality it is enough to show that 
${\bf EC}(S_{i,j})=\cup_{S \in {\bf POSX}_{i,j}} {\bf EC}(S)$.
Indeed, let $S' \in {\bf EC}(S_{i,j})$. As $S'$ is a satisfying assignment
of $\phi_H$, at least one variable of $C_{i,j}$ must occur positively in $S'$.
As all the non-root variables of $C_{i,j}$ occur negatively in $S_{i,j}$ and hence in $S'$,
at least one root variable of $C_{i,j}$ (either $t_i$ or $t_j$) must occur positively in $S'$.  
Therefore, there is $S \in {\bf POSX}_{i,j}$ such that $S \subseteq S'$.
Conversely, if $S' \in \cup_{S \in {\bf POSX}_{i,j}} {\bf EC}(S)$, then $S' \in {\bf EC}(S)$
for some $S \in {\bf POSX}_{i,j}$ and therefore $S \subseteq S'$. Now as $S_{i,j} \subseteq S$, we conclude that $S_{i,j} \subseteq S'$.
$\blacksquare$

The combination of Lemma \ref{lem:guardprob} and Lemma \ref{lem:pos}
allow us to compute $Pr(X_{i,j}=1)$ in terms of $|S_{i,j}|$. It remains to impose an upper bound
on $|S_{i,j}|$.

\begin{lemma} \label{lem:witbound}
For each pseudoedge $\{t_i,t_j\}$ of $H$, $|S_{i,j}| \leq 4*\log m/4.9+5$.
\end{lemma}

{\bf Proof.}
For the purpose of this lemma, the height of a rooted tree $T$ is the largest
number of vertices in a root-leaf path of $T$. We denote by $\ell(T)$ the number
of leaves of $T$.

\begin{claim} \label{cl:balpath}
The height of a full balanced binary tree $T$ is at most $\log(\ell(T))+2$.
\end{claim}

{\bf Proof.}
By definition of pseudoexpander, $|P^{1/2}_{i \rightarrow j}| \leq \log m/4.9+3$.
The variables of $S_{i,j}$ consist of internal variables of 
$P^{1/2}_{i \rightarrow j}$, the internal variables of 
$P^{1/2}_{j \rightarrow i}$, the siblings of these two, and $\ell_{i,j}$.
The number of internal vertices of  $P^{1/2}_{i \rightarrow j}$ 
is at most $\log m/4.9+1$. Clearly, the same upper bound holds for the number
of internal vertices of $P^{1/2}_{j \rightarrow i}$, for the number of siblings
of the internal vertices of $P^{1/2}_{i \rightarrow j}$, and the number of siblings
of the internal vertices of $P^{1/2}_{j \rightarrow i}$. Therefore, the total size
of $S_{i,j}$ is at most $4*(\log m/4.9+1)+1=4*\log m/4.9+5$.
$\blacksquare$

{\bf Proof of the first statement of Theorem \ref{theor:indi}.}
The lower bound on $Pr(X_{i,j}=1)$ is obtained by the following line of reasoning. 
\begin{multline}
Pr(X_{i,j}=1)=\sum_{S \in {\bf POSX}_{i,j}} Pr({\bf EC}(S))=
\sum_{S \in {\bf POSX}_{i,j}} (1/2)^{|S|}=\\3*(1/2)^{|S_{i,j}|+2} \geq
3*(1/2)^{4*\log m/4.9+7}=3/128*m^{-4/4.9}
\end{multline}
where the first equality follows from the second statement of Lemma \ref{lem:pos},
the second equality follows from the combination of the first statement of Lemma \ref{lem:pos}
and Lemma \ref{lem:guardprob}, the third equality follows from definition of ${\bf POSX}_{i,j}$,
the first inequality follows from Lemma \ref{lem:witbound}.
$\blacksquare$

In order to establish the second statement of Theorem \ref{theor:indi},
we represent $Pr(X_{i,j}=0)$ as the sum of probabilities 
of ${\bf EC}$-events of guarded sets of literals.
For this purpose, we define
the set ${\bf NEGX}_{i,j}$ consisting of all the set of literals
over the variables of $S_{i,j}$ plus $\{t_i,t_j\}$ 
such that each $S \in {\bf NEGX}_{i,j}$
does not falsify any clause and satisfies one of the following two conditions.
\begin{enumerate}
\item At least one non-root variable of $C_{i,j}$ is assigned positively.
\item At least one sibling of an internal variable of $C_{i,j}$ is assigned negatively.
\end{enumerate}

\begin{lemma} \label{lem:neg}
\begin{enumerate}
\item Each $S \in {\bf NEGX}_{i,j}$ is guarded.
\item $Pr(X_{i,j}=0)=Pr(\bigcup_{S \in {\bf NEGX}_{i,j}} {\bf EC}(S))$.
\end{enumerate}
\end{lemma}

{\bf Proof.}
$S \in {\bf NEGX}_{i,j}$ assigns all the variables of $C_{i,j}$ and does
not falsify $C_{i,j}$. This means that either one of non-leaf variables
of $C_{i,j}$ occur positively in $S$ (and in this case the second condition
of Definition \ref{def:guard} is satisfied) or, otherwise, the $\leaf_{i,j}$
occur positively in $C_{i,j}$ and the rest of variables occur negatively,
thus satisfying the third condition of Definition \ref{def:guard}.
Consequently, $S$ is guarded from $C_{i,j}$.
Except $\leaf_{i,j}$, no other leaf variable occurs in $S$, therefore
$S$ is guarded from the rest of clauses of $\phi_H$ according to the first
condition of Definition \ref{def:guard}, confirming the first statement.

For the second statement, it is enough to show that 
$\{S'| S' \in {\bf SAT}(H) \wedge X_{i,j}(S')=0\}=\bigcup_{S \in {\bf NEGX}_{i,j}} {\bf EC}(S)$.

Let $S' \in {\bf SAT}(H)$ such that
$X_{i,j}(S')=0$. By Observation \ref{obs:ij}, $S_{i,j} \nsubseteq S'$.
Then the projection $S''$ of $S'$ to the variables of $S_{i,j}$ plus $\{t_i,t_j\}$
does not contain $S_{i,j}$ either. Clearly, $S''$ satisfies one of the two
conditions for ${\bf NEGX}_{i,j}$ and, as $S''$ does not falsify any clause, 
$S'' \in {\bf NEGX}_{i,j}$. As $S' \in {\bf EC}(S'')$, we conclude that 
$S' \in \bigcup_{S \in {\bf NEGX}_{i,j}} {\bf EC}(S)$.

Conversely, let $S' \in \bigcup_{S \in {\bf NEGX}_{i,j}} {\bf EC}(S)$.
Then there is $S'' \in {\bf NEGX}_{i,j}$ such that $S' \in {\bf EC}(S'')$.
By the definition of ${\bf NEGX}_{i,j}$, all the variables of $S_{i,j}$ occur 
in $S''$ and at least one of them occurs oppositely to its occurrence in $S_{i,j}$.
Clearly, the same is preserved in $S'$. Hence, $S_{i,j} \nsubseteq S'$ and therefore
$X_{i,j}(S')=0$.
$\blacksquare$

We also need one more statement on guarded sets which will be proved in Subsection \ref{sec:guard}
\begin{lemma} \label{lem:guardtransf}
Let ${\bf S_1}, \dots, {\bf S_q}$
be guarded sets of literals over sets $V_1, \dots, V_q$
of variables such that $V_i \cap V_j=\emptyset$
for all $1 \leq i,j \leq q$.
Then $Pr(\bigcap_{i=1}^q \bigcup_{S \in {\bf S}_i} {\bf EC}(S))=
      \prod_{i=1}^q Pr(\bigcup_{S \in {\bf S}_i} {\bf EC}(S))$
\end{lemma}

{\bf Proof of the second statement of Theorem \ref{theor:indi}.}
We need to show that for any $M' \subseteq M$
and for any set of numbers $a_{i,j} \in \{0,1\}$ for each $\{t_i,t_j\} \in M'$,
$Pr(\bigcap_{\{t_i,t_j\} \in M'} (X_{i,j}=a_{i,j}))=\prod_{\{t_i,t_j\} \in M'} Pr(X_{i,j}=a_{i,j})$.
Partition $M'$ into $M'_1$ and $M'_2$ such that $a_{i,j}=1$ for all $\{t_i,t_j\} \in M'_1$
and $a_{i,j}=0$ for all $\{t_i,t_j\} \in M'_2$. 
Then
\begin{multline}
Pr(\bigcap_{\{t_i,t_j\} \in M'} (X_{i,j}=a_{i,j}))=
Pr(\bigcap_{\{t_i,t_j\} \in M'_1} (X_{i,j}=1) \cap \bigcap_{\{t_i,t_j\} \in M'_2} (X_{i,j}=0))=\\
Pr(\bigcap_{\{t_i,t_j\} \in M'_1} \bigcup_{S \in {\bf POSX}_{i,j}} {\bf EC}(S) \cap
   \bigcap_{\{t_i,t_j\} \in M'_1} \bigcup_{S \in {\bf NEGX}_{i,j}} {\bf EC}(S))=\\
\prod_{\{t_i,t_j\} \in M'_1} Pr(\bigcup_{S \in {\bf POSX}_{i,j}} {\bf EC}(S))*
\prod_{\{t_i,t_j\} \in M'_2} Pr(\bigcup_{S \in {\bf NEGX}_{i,j}} {\bf EC}(S))=\\
\prod_{\{t_i,t_j\} \in M'_1} Pr(X_{i,j}=1)*
\prod_{\{t_i,t_j\} \in M'_2} Pr(X_{i,j}=0)=
\prod_{\{t_i,t_j\} \in M'} Pr(X_{i,j}=a_{i,j})
\end{multline}
where the second and the fourth equalities follow from the combination Lemma \ref{lem:pos}
and Lemma \ref{lem:neg},
the third equality follows from Lemma \ref{lem:guardtransf} together with Lemma \ref{lem:pos} and 
Lemma \ref{lem:neg}.
$\blacksquare$

\subsection{Proof of Lemmas for guarded sets} \label{sec:guard}
\begin{lemma} \label{lem:fixfacts}
Let $S$ be a guarded set of literals. Then the following statements hold.
\begin{enumerate}
\item Let $x$ be a literal of a non-leaf variable that does not occur in $S$.
Then $S \cup \{x\}$ is a guarded set if literals with $Fix(S \cup \{x\})=Fix(S)$.
\item Let $x$ be a literal of $\leaf_{i,j}$ such that $\leaf_{i,j}$ does not occur
in $S$ and at least one variable of $C_{i,j}$ occurs positively in $S$.

Then $S \cup \{x\}$ is a guarded set if literals with $Fix(S \cup \{x\})=Fix(S)$.
\item Assume that $\leaf_{i,j}$ does not occur in $S$ and the rest of variables of
$C_{i,j}$ occur negatively in $S$. Then $S \cup \{\leaf_{i,j}\}$ is 
a guarded set if literals with $Fix(S \cup \{\leaf_{i,j}\})=Fix(S) \cup \{\leaf_{i,j}\}$.
\end{enumerate}
\end{lemma}

{\bf Proof.}
Suppose that $x \notin S$ is a literal of a variable $y$ 
such that $S \cup \{x\}$ is not guarded while $S$ is guarded. 
Hence, there is a clause $C_{i,j}$ such that
$S$ is guarded from $C_{i,j}$ while $S \cup \{x\}$ is not guarded
from $C_{i,j}$. Clearly, this means that $y$ occurs in $C_{i,j}$,
and hence not all variables of $C_{i,j}$ occur in $S$
Consequently, $S$ does not satisfy the third condition of Definition \ref{def:guard} regarding
$C_{i,j}$ because this condition clearly requires that all the variables of $C_{i,j}$ occur in $S$.
Also, note that if $S$ satisfies the second condition of Definition \ref{def:guard} 
regarding $C_{i,j}$, this condition cannot be violated by introduction of a new literal.
Consequently, $S$ can satisfy the first condition only, that is some non-leaf
variables of $C_{i,j}$ occur in $S$ and all their occurrences are \emph{negative}.
The first condition clearly cannot be violated if $y$ is a non-leaf variable.
Thus $y=\leaf_{i,j}$. 

The reasoning in the previous paragraph immediately implies that $S \cup \{x\}$,
as in the first statement, is guarded because $x$ is a literal of an internal variable
and that $S \cup \{x\}$, as in the second statement, is guarded because  the second 
condition is satisfied regarding $C_{i,j}$. For the third statement, note 
since \emph{all} the non-leaf variables of $C_{i,j}$ occur negatively in $S$,
$S \cup \{\leaf_{i,j}\}$ satisfies the third condition of Definition \ref{def:guard}
regarding $C_{i,j}$. Therefore, $S \cup \{\leaf_{i,j}\}$ remains guarded. 

Now, let us prove the equations concerning the fixed sets.
Observe that in all the considered cases $Fix(S) \subseteq Fix(S \cup \{x\})$.
Indeed, suppose $\leaf_{i,j} \in Fix(S)$. This means that $\leaf_{i,j} \in S$
and the rest of variables of $C_{i,j}$ occur negatively in $S$. Clearly,
these occurrences are preserved in any superset of $S$.

Now, suppose that $\leaf_{i,j} \in Fix(S \cup \{x\}) \setminus Fix(S)$, and hence $x=\leaf_{i,j}$. 

Then $\leaf_{i,j} \in (S \cup \{x\}) \setminus S$. Indeed, if $\leaf_{i,j} \in S$
then to keep $S$ guarded from $C_{i,j}$, the second condition of Definition
\ref{def:guard} must be satisfied (the third condition cannot be satisfied because
of the assumption that $\leaf_{i,j} \notin Fix(S)$). But then one of non-leaf 
variables of $C_{i,j}$ occurs in $S$ (and hence also in $S \cup \{x\}$) positively
blocking the possibility of $\leaf_{i,j}$ being fixed. Moreover, all the internal
variables of $C_{i,j}$ occur in $S$ negatively. This immediately implies that
$Fix(S \cup \{x\}) \subseteq Fix(S)$ in the first case (because $x$ is not a leaf variable)
and in the second case (because $S$ contains a positive occurrence of a non-leaf variable
of $C_{i,j}$).

For the third case, the above reasoning implies
that $Fix(S \cup \{x\}) \subseteq Fix(S \cup \{\leaf_{i,j}\})$ and 
since by definition, $\leaf_{i,j} \in Fix(S \cup \{x\})$,
the opposite containment follows as well.
$\blacksquare$ 

{\bf Proof of Lemma \ref{lem:guardprob}.}
The proof is by induction on the number of variables that \emph{do not} occur in $S$.
Assume first that all the variables are assigned.
Note that $S$ does not falsify any clause. Indeed, if clause $C_{i,j}$
is falsified then none of the conditions of Definition \ref{def:guard}
can be met. By definition of the probability space, 
$Pr({\bf EC}(S))=2^{-|V(H) \setminus Fix(S)|}=2^{-|S \setminus Fix(S)|}$.

Assume now that $S$ does not assign all the variables.
Suppose first that some non-leaf variable $y$ does not occur in
$S$. 

As in any $S' \in {\bf EC}(S)$ $y$ occurs either positively
or negatively, ${\bf EC}(S)={\bf EC}(S \cup \{y\}) \cup {\bf EC}(S \cup \{\neg y\})$.
In fact, by Lemma \ref{lem:disjcont}, the union is disjoint.
Also, note that for each $x \in \{y, \neg y\}$, 
the following is true. 

\begin{multline} \label{eq:minortransf}
|S \cup \{x\} \setminus Fix(S \cup \{x\})|=|S \cup \{x\}|-|Fix(S \cup \{x\})|=\\
|S+1|-|Fix(S)|=|S \setminus Fix(S)|+1
\end{multline}
where the first equality is correct because $Fix(S \cup \{x\}) \subseteq S \cup \{x\}$,
$|S \cup \{x\}|=|S|+1$ since $x$ does not occur in $S$, $Fix(S \cup \{x\})=Fix(S)$
according to the first statement of Lemma \ref{lem:fixfacts}, and the last equality
is correct because $Fix(S) \subseteq S$.

Therefore,
\begin{multline}
Pr({\bf EC}(S))=Pr({\bf EC}(S \cup \{y\}))+Pr({\bf EC}(S \cup \{\neg y\}))=\\
2^{-|(S \cup \{y\}) \setminus Fix(S \cup \{y\}|}+
2^{-|(S \cup \{\neg y\}) \setminus Fix(S \cup \{\neg y\}|}=
2^{-(|S \setminus Fix(S)|+1)}+2^{-(|S \setminus Fix(S)|+1)}=\\
2^{-|S \setminus Fix(S)|} 
\end{multline}
where the second equality follows because, according to the first statement of 
Lemma \ref{lem:fixfacts}, both $S \cup \{y\}$ and $S \cup \{\neg y\}$ are
guarded and they assigned more variables than $S$, therefore the statement 
of this lemma holds for them according to the induction assumption.
The third equality follows from \eqref{eq:minortransf}.

Suppose now that $S$ assigns all the non-leaf variables and let $\leaf_{i,j}$
be a variable not assigned by $S$. If at least one of variables of $C_{i,j}$
occurs positively in $S$, then the reasoning is analogous to the above with
the second statement of Lemma \ref{lem:fixfacts} used instead of the first one.

Otherwise, all the non-leaf variables of $C_{i,j}$ occur in $S$ negatively.
In this case $\leaf_{i,j}$ occurs positively in any element of ${\bf EC}(S)$
(for otherwise $C_{i,j}$ is falsified). That is, ${\bf EC}(S)={\bf EC}(S \cup \{\leaf_{i,j}\})$.
Therefore
\begin{multline}
Pr({\bf EC}(S))=Pr({\bf EC}(S \cup \{\leaf_{i,j}\}))=
(1/2)^{|(S \cup \{\leaf_{i,j}\}) \setminus Fix(S \cup \{\leaf_{i,j}\})|}=\\
(1/2)^{|(S \cup \{\leaf_{i,j}\})|-|Fix(S \cup \{\leaf_{i,j}\})|}=
(1/2)^{|S|+1-(|Fix(S)|+1)}=\\
(1/2)^{|S|-|Fix(S)|}=(1/2)^{|S \setminus Fix(S)|}
\end{multline}
where the second equality follows from the third statement of Lemma \ref{lem:fixfacts}
and the induction assumption, the third equality follows because 
$Fix(S \cup \{\leaf_{i,j}\}) \subseteq S \cup \{\leaf_{i,j}\}$.
For the fourth equality, note that $|(S \cup \{\leaf_{i,j}\})|=|S|+1$  
as $\leaf_{i,j} \notin S$ and that by the third statement of Lemma \ref{lem:fixfacts},
$Fix(S \cup \{\leaf_{i,j}\})=Fix(S) \cup \{\leaf_{i,j}\}$ and hence
$|Fix(S \cup \{\leaf_{i,j}\})|=|Fix(S)|+1$. The last equality follows 
because $Fix(S) \subseteq S$.
$\blacksquare$

\begin{lemma} \label{lem:guardcompose}
Let $S_1, \dots, S_q$ be guarded sets of literals over pairwise disjoint sets of variables.
Then $Pr({\bf EC}(\bigcup_{i=1}^q S_i))=\prod_{i=1}^q Pr({\bf EC}(S_i))$.
\end{lemma}

{\bf Proof.}
Let us show that $S=\bigcup_{i=1}^q S_i$ is guarded.
Indeed, consider a clause $C_{j,k}$. If $\leaf_{j,k}$
does not occur in $S$, then $S$ satisfies the first 
condition of Definition \ref{def:guard}.
Otherwise, there is $S_i$ such that $\leaf_{j,k}$ occurs in
$S_i$. As $S_i$ is guarded from $C_{j,k}$ either the second or
the third condition is satisfied for $S_i$ regarding $C_{j,k}$.
It is not hard to see that both these conditions are preserved in every
superset of $S_i$.

By Lemma \ref{lem:guardprob} and also taking into account that 
$Fix(S) \subseteq S$ and that $S_1, \dots, S_q$ are mutually disjoint,
we obtain the following.
\begin{multline} \label{eq:comp1}
Pr({\bf EC}(S))=(1/2)^{|S \setminus Fix(S)|}=(1/2)^{|S|-|Fix(S)|}=\\
(1/2)^{(\sum_{i=1}^q |S_i|)-|Fix(S)|}
\end{multline}

We next show that $Fix(S)=\bigcup_{i=1}^q Fix(S_i)$.
As each fixed element of $F_i$ remains fixed in a superset of $F_i$,
$\bigcup_{i=1}^q Fix(S_i) \subseteq Fix(S)$.

Conversely, suppose that there is $\leaf_{j,k} \in Fix(S) \setminus \bigcup_{i=1}^q Fix(S_i)$.
As $\leaf_{j,k} \in S$, there is $i$ such that $\leaf_{j,k} \in S_i$.
Since $\leaf_{j,k} \notin Fix(S_i)$, $S_i$ does not satisfy the third condition
of Definition \ref{def:guard} regarding $C_{j,k}$. 
To remain guarded from $C_{j,k}$,
$S_i$ must satisfy the second condition of Definition \ref{def:guard} regarding
$C_{j,k}$. 
That is, a non-leaf variable of $C_{j,k}$ occurs positively in $S_i$ and hence also in 
$S$, in contradiction to $\leaf_{j,k} \in Fix(S)$, thus confirming that 
$Fix(S) \subseteq \bigcup_{i=1}^q Fix(S_i)$.

Due to the pairwise disjointness of $S_1, \dots, S_q$, $Fix(S_1), \dots, Fix(S_q)$
are pairwise disjoint (as being subsets of $S_1, \dots, S_q$ respectively.)
Therefore, $|Fix(S)|=\sum_{i=1}^q |Fix(S_i)|$. Substituting this equality into the
right-most item of \eqref{eq:comp1}, we obtain

\begin{multline}
Pr({\bf EC}(S))=(1/2)^{(\sum_{i=1}^q |S_i|)-\sum_{i=1}^q |Fix(S_i)|}=
(1/2)^{\sum_{i=1}^q (|S_i|-|Fix(S_i)|)}=\\(1/2)^{\sum_{i=1}^q |S_i \setminus Fix(S_i)|}=
\prod_{i=1}^q (1/2)^{|S_i \setminus Fix(S_i)|}=\prod_{i=1}^q Pr({\bf EC}(S_i))
\end{multline}
where the second equality follows because $Fix(S_i) \subseteq S_i$ and the last equality
follows from Lemma \ref{lem:guardprob}.
$\blacksquare$

{\bf Proof of Lemma \ref{lem:guardtransf}.}
The lemma is proved through the following line of reasoning.
\begin{multline}
Pr(\bigcap_{i=1}^q \bigcup_{S_i \in {\bf S}_i} {\bf EC}(S_i))=
Pr(\bigcup_{(S_1, \dots, S_q) \in {\bf S}_1 \times \dots \times {\bf S}_q} \bigcap_{i=1}^q {\bf EC}(S_i))=\\
Pr(\bigcup_{(S_1, \dots, S_q) \in {\bf S}_1 \times \dots \times {\bf S}_q} {\bf EC}(\bigcup_{i=1}^q S_i))=
\sum_{(S_1, \dots, S_q) \in {\bf S}_1 \times \dots \times {\bf S}_q} Pr({\bf EC}(\bigcup_{i=1}^q S_i))=\\
\sum_{(S_1, \dots, S_q) \in {\bf S}_1 \times \dots \times {\bf S}_q} \prod_{i=1}^q Pr({\bf EC}(S_i))=
\prod_{i=1}^q \sum_{S_i \in {\bf S}_i} Pr({\bf EC}(S_i))=\\
\prod_{i=1}^q Pr(\bigcup_{S_i \in {\bf S}_i} {\bf EC}(S_i))
\end{multline}

For the second equality, we demonstrate that for any set $S_1, \dots, S_q$ of literals,
${\bf EC}(\bigcup_{i=1}^q S_i)=\bigcap_{i=1}^q {\bf EC}(S_i)$. 
If $S \in {\bf EC}(\bigcup_{i=1}^q S_i)$ then $S_i \subseteq S$ for $1 \leq i \leq q$,
thus $S \in {\bf EC}(S_i)$ for $1 \leq i \leq q$. 
Conversely,
if $S \in \bigcap_{i=1}^q {\bf EC}(S_i)$, then $S_i \subseteq S$ for all $1 \leq i \leq q$.
That is $\bigcup_{i=1}^q S_i \subseteq S$ and hence $S \in {\bf EC}(\bigcup_{i=1}^q S_i)$.
Note that since the elements of \emph{both} sets are subsets
of ${\bf SAT}(H)$, there is no need to explicitly mention that the $S$ being considered
is a satisfying assignment of $\phi_H$.  

The third equality follows from Lemma \ref{lem:disjcont} because for two distinct 
tuples $(S_1, \dots, S_q)$, $\bigcup_{i=1}^q S_i$ are distinct sets of literals over
the same set of variables. The fourth equality follows from Lemma \ref{lem:guardcompose}.
The last equality follows from Lemma \ref{lem:disjcont} applied individually to
each ${\bf S}_i$. 
$\blacksquare$

\appendix
\section{Simulation of a $d$-{\sc nbp} by a uniform $d$-{\sc nbp}}
Let $Z$ be a $d$-{\sc nbp} and let $v_1, \dots, v_n$ be its vertices listed in the topological order.
For a variable $x$ of $F(Z)$ and a node $v$ of $Z$, denote by $n_Z(v,x)$ the largest number of occurrences 
of variable $x$ on a path of $Z$ from the source to $v$.
We introduce a sequence of {\sc nbp}s $Z_1, \dots, Z_n$ defined as follows.
$Z_1=Z$. 

For $i>1$, $Z_i=Z_{i-1}$ if $v$ has only one in-coming edge.
Otherwise, $Z_{i-1}$ is transformed into $Z_i$ by the following process.
First, whenever there is an edge $(u,v)$ of $Z_{i-1}$ that is labelled subdivide
this edge by introducing a new vertex $w$ and replacing $(u,v)$ with edges
$(u,w)$ and $(w,v)$ with $(u,w)$ being labelled by the label of $(u,v)$ and
$(w,v)$ being unlabelled. This way we obtain an {\sc nbp} $Z'_i$ where $v$
has the same number of in-coming edges as in $Z_{i-1}$ however all in-coming
edges of $v$ are unlabelled.

The second stage of transformation ensures that 
for each incoming edge $(u,v)$ of $Z'_i$ and for each variable $x$,
the largest number of occurrences of $x$ on a path from the sink to $v$ that
goes through $(u,v)$ is $n_{Z_{i-1}}(v,x)$.
To this end, for each edge $(u,v)$ of $Z'_i$, let $x_1, \dots, x_q$ be the set of variables
such that for each $1 \leq i \leq q$, $n_{Z'_{i}}(u,x_i)<n_{Z_{i-1}}(v,x_i)$.
Let $m_1, \dots, m_q$ be such that $m_i=n_{Z_{i-1}}(v,x_i)-n_{Z'_i}(u,x_i)$ for $ 1 \leq i \leq q$.
Introduce new vertices $v_{1,1}, \dots, v_{1,m_1}, \dots v_{q,1}, \dots, v_{q,m_q}$.
Replace the edge $(u,v)$ by the following edges.
\begin{itemize}
\item Two edges from $u$ to $v_{1,1}$ one labelled with $x_1$, the other labelled with $\neg x_1$.
\item For $1<j \leq m_1$, two edges from $v_{1,j-1}$ to $v_{1,j}$, 
one labelled with $x_1$, the other labelled with $\neg x_1$.
\item For each $1<i \leq q$, two edges from $v_{i-1,m_{i-1}}$ to $v_{i,1}$ one labelled
with $x_i$, the other labelled with $\neg x_i$.
\item For each $1 <i \leq q$ ad each $1<j \leq m_i$, two edges from $v_{i,j-1}$ to $v_{i,j}$
one labelled with $x_i$, the other labelled with $\neg x_i$.
\item Two edges from $v_{q,m_q}$ to $v$ one labelled with $x_q$, the other labelled with $\neg x_q$.
\end{itemize}

Arguing by induction on $v_1, \dots, v_n$, the following statements can be made.
\begin{enumerate}
\item Each $Z_i$ is a $d$-{\sc nbp}.
\item For each $Z_i$ and each variable $x$, all the paths from the source to $v_i$
have the same number of occurrences of $x$ which is $n_{Z_{i-1}}(v_i,x)$.
\item All the $Z_i$ compute the same function.
\end{enumerate}

For $i=1$, all three statements follow from definition of $Z$ so we assume $i>1$.

Next, we are going to show that for each in-neighbour $u$ of $v$ in $Z'_i$ and for each variable
$x$, the number of occurrences of $x$ on each path of $Z'_i$ from the soruce to $u$ is exactly
$n_{Z'_i}(u,x_i)$. 

Indeed, assume first that $u$ is a vertex of $Z_{i-1}$ (that is, it has
not been introduced in $Z'_i$ as a result of subdivision). Then $u=v_j$ for some $j<i$.
Indeed, otherwise, $u$ is a new vertex that has been introduced during construction of some $Z_j$
for $j<i$. Then only vertex of $Z_j$ such a new vertex can be an in-neighbour at the stage of its introducing
is $v_j$. Moreover, such a vertex cannot have an out-going edge to $v_i$ introduced at a latter stage 
because one end of such edge must be a new vertex introduced at a latter stage (which is not $u$
that has been introduced during construction of $Z_j$ and not $v_i$ that exists in $Z_1$).
Thus, in this case, the desired statement is correct by the induction assumption.
The last statement requires further explanation. The induction assumption, as such,
states that the number of occurrences of $x$ on each path of $Z_j$ from the source to
$v_j$ is $n_{Z_j}(v_j,x)$. In particular, the induction assumption says nothing about the paths
from the surce to $v_j$ in $Z_{i'}$ for $i'>j$. 
However, the paths from the source to $v_j$ 
are not affected by the subsequence transformations of $Z_{j}$ into $Z_{j+1}$ and so on, hence
the induction assumption remains invariant and, moreover, $n_{Z_j}(v_j,x)=n_{Z_{i'}}(v_j,x)$
for all $i'>j$. As paths from the source to $u$ are not affected by a transformation
from $Z_{i-1}$ to $Z'_i$, $n_{Z_j}(v_j,x)=n_{Z'_i}(v_j,x)$ as required. 

Otherwise, if $u$ is not a vertex of $Z_{i-1}$, 
there is an in-neighbour $u'$ of $v_i$ such that $u$ is introduced as a result of
subdivision of an edge $e$ from $u'$ to $v_i$. By the previous paragraph $u'=v_j$ for some 
$j<i$. Then, by the induction assumption, the number of occurrences of $x$ on any path
of $Z'_i$ from the source to $u'$ is $n_{Z_{i-1}}(v_j,x)$ (the reasoning is as in the previous
paragraph). Any path of $Z'_i$ from the source to $u$ is a path from the source to $u'$ plus
the edge $(u',u)$. If the variable labelling $(u',u)$ is not $x$ then the number of occurrences
on any such path is $n_{Z_{i-1}}(v_j,x)$, otherwise it is $n_{Z_{i-1}}(v_j,x)+1$. In both cases,
the number of occurrences on any path of $Z'_i$ from the source to $u$ is invariant and hence
equals $n_{Z'_i}(u,x)$.

Note that $n_{Z'_i}(u,x) \leq n_{Z_{i-1}}(v_i,x)$ because it is just the largest number of 
occurrences of $x$ on a path of $Z_{i-1}$ from the source to $v_i$ going through
a particular edge ($(u,v)$ or $(u',v)$ depending on whether this edge is subdividied
by transformation from $Z_{i-1}$ to $Z'_i$). If $n_{Z'_i}(u,x)=n_{Z_{i-1}}(v_i)$ 
then $x$ does not occur on a path between $u$ and $v_i$ in $Z_i$. Therefore
the number of occurrences of $x$ on a path of $Z_i$ from the soruce to $v_i$ is
as in $Z'_i$ which is $n_{Z'_i}(u,x)=n_{Z_{i-1}}(v_i,x)$.
Otherwise, there is $1 \leq j \leq q$ such that $x=x_j$ and on each path
of $Z_i$ from $u$ to $v$ there are $m_j=n_{Z_{i-1}}(v_i,x)-n_{Z'_i}(u,x)$ occurrences of $x$.
That is the total number of occurrences on each path of $Z_i$ from the source to $v_i$ passing
through $u$ is $n_{Z'_i}(u,x)+m_j=n_{Z_{i-1}}(v_i,x)$. Since $u$ is selected as an arbitrary in-neighbour
of $v_i$ in $Z'_i$, there are $n_{Z_{i-1}}(v_i,x)$ occurrences of $x$ on any path from the source to $v_i$.
This confirms the second statement. 

For the first statement, it is sufficient to verify that on any source-sink path of $Z_i$ going through
$v_i$ each variable $x$ occurs at most $d$ times. Indeed, on each $v_i$-sink path $Q$ of 
$Z_{i-1}$ each variable $x$ occurs at most $d-n_{Z_{i-1}}(v_i,x)$ times. Otherwise, take
a path $P$ of $Z_{i-1}$ from the source to $v_i$ on which there are $n_{Z_{i-1}}(v_i,x)$
occurrences of $Z_i$ (such a path exists by definition of $n_{Z_{i-1}}(v_i,x)$) and let
$Q$ be a path of $Z_{i-1}$ from $v_i$ to the sink having more that $d-n_{Z_{i-1}}(v_i,x)$
occurrences of $x$. Then $P+Q$ is source-sink path of $Z_{i-1}$ with more than $d$ occurrences 
of $x$ in contradiction to the first statement holding regarding $Z_{i-1}$ by the induction
assumption. Since the paths from $v_i$ to sink are not affected by transformation from
$Z_{i-1}$ to $Z_i$, wee conclude that on each $v_i$-sink path $Q$ of 
$Z_{i}$, each variable $x$ occurs at most $d-n_{Z_{i-1}}(v_i,x)$ times.
Together with the second statement, this implies that on each source sink path of $Z_i$ going 
through $v_i$ each variable occurs at most $d$ times thus confirming the first statement.

Let us verify the third statement, in particular, let us show that $Z_{i-1}$ and $Z_i$ 
compute he same function. Let $S$ be a satisfying assignment of $F(Z_{i-1})$. This means
that $Z_{i-1}$ has a computational path $P$ with $A(P) \subseteq S$. If $P$ does not contain
$v_i$ then $P$ is a computational path of $Z_i$ and hence $S$ is a satisfying assignment
of $Z_i$. Otherwise, let $u$ be the vertex preceding $v_i$ in $P$. If $u$ is an in-neighbour
of $v_i$ in $Z_i$ then, again, $P$ is a computational path of $Z_i$. Otherwise, let 
$P_1$ be the prefix of $P$ ending at $u$ and let $P_2$ be the suffix of $P$ beginning at $v_i$.
By contruction, both $P_1$ and $P_2$ are paths of $Z_i$.
Then choose a path $P_0$ between $u$ and $v$ in $Z_i$ so that each variable labelling 
the path has the same occurrence as in $S$ (this possible to do due to the `paraller edges'
construction as described above). Then $P_1+P_0+P_2$ is a path of $Z_i$ such that
$A(P_1+P_0+P_2) \subseteq S$. 

Conversely, let $S$ be a satisfying assignment of $Z_i$ and let $P$ be a computational
path of $Z_i$ with $A(P) \subseteq S$. It is sufficient to consider the case
where $P$ is not a computational path of $Z_{i-1}$.
Then $P$ contains new vertices of $Z_i$ and hence contains $v_i$
(because any path from a new vertex to the sink goes through $v_i$).
Moreover, the immediate predecessor of $v_i$ in $P$ is a new vertex of $Z_i$.
Let $u$ be the last in-neighbour of $v_i$ in $Z_{i-1}$ preceding $v_i$ in $P$. 
Let $P'$ be obtained from $P$ by replacing the $u \longrightarrow v_i$ subpath
of $P$ with the $(u,v_i)$ edge of $Z_{i-1}$ carrying the same label as in $Z_{i-1}$.
Then $P'$ is a computational path 
By construction $A(P') \subseteq A(P)$ (in particular, the process of transformation of edge
$(u,v_i)$ into a path retains the label of $(u,v_i)$ on any resulting $u \longrightarrow v_i$
path and hence $A(P') \subseteq S$

It follows that $Z_n$ computes the same function as $Z$ and each variable $x$ occurs the
same number of times $n_x \leq k$ on each computational path of $Z_n$.
For those variables $x$ where $n_x<k$ add $k-n_x$ entries on each in-coming edge of the
sink using the same subdivision technique as was used for transformation from $Z_{i-1}$ to $Z_i$.
Then, by the analogous reasoning it can be verified that the resulting {\sc nbp} $Z^*$ becomes 
uniform $k$-{\sc nbp} computing the same function as $Z$.

Let us calculate the size of $Z^*$ in terms of the size of $Z$.
Note that for the purpose of this proof $|Z|$ is denoted by $n$.
The number of variables cannot be more than the number of edges of $Z$
therefore is bounded by $O(n^2)$. 
On each of the $n$ iterations, at most $n$ edges are transformed
(an in-degree of a node cannot be larger than $n$) and on each node
the number of copies added is bounded by the number of variables multiplied
by $k$ that is $O(n^2)*k$. 
Therefore, the size of $Z^*$ is $O(n^4*k)$. 
(Note that the size increase of the final transformation from $Z_n$ to $Z^*$ is dominated by this
complexity and so can be safely ignored).

\section{Proof of Theorem \ref{nbpfinal}} \label{sec:detconst}
\subsection{Constructively created CNFs} \label{sec:mainconst}
We will next define a set of graphs $G_m$.
Then we will define a set of {\sc cnf}s $\varphi(G_m)$ associated with graphs $G_m$.
The {\sc cnf}s $\varphi(G_m)$ is the class for which we will prove the
lower bound stated in the theorem.

Let $T^*_1, \dots, T^*_m$ be a set binary rooted trees with $\lfloor m^{1/4.9} \rfloor$ leaves each
and height at most $\log m/4.9+2$. Such trees can be created as follows.
Let $a \geq m^{1/4.9}$ be the nearest to $m^{1/4.9}$ integer which is a power of two.
Take a complete binary tree $T'$ with $a$ leaves. Then choose an arbitrary set $L$ of 
$\lfloor m^{1/4.9} \rfloor$ leaves and let each $T^*_i$ be obtained from $T'$
by taking the union of all root-leaf trees ending at $L$.
The height of $T'$ is $\log a+1$ and, as $a \leq 2*m^{1/4.9}$, the height
does not exceed $\log m/4.9+2$.

Then for each $1 \leq i \neq j \leq m$ and each pair of leaves $\ell_1$ of $T_i$
and $\ell_2$ of $T_j$, introduce a new vertex $v$ and connect it by edges to $\ell_1$
and $\ell_2$. We call these new vertices $v$ \emph{subdivision vertices}.
The graph $G_m$ consists of the union of $T^*_1, \dots T^*_m$ plus the subdivision
vertices together with the edges incident to them.

The variables of $\varphi(G_m)$ are the vertices of $G_m$.
The clauses correspond to the subdivision vertices.
Let $v$ be a subdivision vertex. Then the clause $C_v$ corresponding to it
is created as follows. Let $\ell_1$ and $\ell_2$ be two leaves of trees
$T^*_i$ and $T^*_j$ incident to $v$. Let $P_1$ and $P_2$ be respective root-leaf paths
of $T^*_i$ and $T^*_j$ ending with $\ell_1$ and $\ell_2$.
The literals of $C_v$ are $V(P_1) \cup \{v\} \cup V(P_2)$.

It is not hard to see that $\varphi(G_m)$ can be created by a deterministic procedure
taking polynomial time in the number of variables.

In the next subsection we formally prove that for a sufficiently large $m$,
it is possible to assign a subset of subdivision vertices with $true$ so that
the remaining {\sc cnf} is $\phi_H$ where $H$ is a \pex\ with $m$ roots.
Then, in Section B.3. we will show that the lower bound for $\phi_H$ implies
that lower bound for $\varphi(G_m)$.

\subsection{Extraction of {\sc cnf}s based on pseudoexpanders}
Let $F$ be a Boolean function and let $S$ be a partial assignment to a subset
of its variables. Then $F|_S$ is a function on variables not assigned by 
$S$ and $S'$ is a satisfying assignment of $F|_S$ if and only if $S \cup S'$
is a satisfying assignment of $F$.

If $\varphi$ is a monotone {\sc cnf} and $S$ consists of positive literals of
a subset of variables of $\varphi$ then $\varphi|_S$ is obtained by removal
of clauses where the variables of $S$ occur. In this section we are going to
prove the following lemma.

\begin{lemma} \label{lem:redpe}
Let $m$ be an integer and suppose that there exist a pseudoexpander $H$ with $m$
roots and max pseudodegree $m^{1/4.9}$. 
Then there is a set $S$ of positive literals of subdivision vertices of $\varphi(G_m)$
such that $\varphi(G_m)|_S$ is $\phi_{H'}$ where $H'$ is a subgraph of $G_m$ which is
a pseudoexpander with $m$ roots and with $U(H)$ isomorphic to $U(H')$. 
\end{lemma}

{\bf Proof.}
Let ${\bf Roots}(H)=\{t_1, \dots, t_m\}$ and let $t'_1, \dots, t'_m$
be the respective roots of the trees $T^*_1, \dots T^*_m$ of $G_m$ (let 
us call them the \emph{roots} of $G_m$).
Recall that by $U(H)$ we denote the underlying graph of $H$.
For $1 \leq i \leq m$, denote by $I_i$ the set of $j$ such that
$t_i$ is adjacent to $t_j$ in $U(H)$.
Note that since the max-degree of $U(H)$ is at most $m^{1/4.9}$,
the number of leaves of $T^*_i$ is at least $|I_i|$.
For each $i$, fix an arbitrary subset of $|I_i|$ leaves of $T^*_i$
and mark them with with the indices of $I_i$, each index marking one
leaf.

Now consider a subdivision vertex $v$ of $G_m$ connecting a leaf $\ell_1$
of $T_i$ and a leaf $\ell_2$ of $T_j$. Let us call $v$ \emph{meaningful}
if $\ell_1$ is marked with $j$ and $T_j$ is marked with $i$.
Let $S$ be the set consisting of all non-meaningful subdivision vertices
and let $\varphi'=\varphi(G_m)|_S$.

For $1 \leq i \leq m$, let $T'_i$ be the rooted tree obtained from
$T^*_i$ by taking the union of all root-leaf paths of $T^*_i$ that
end with marked leaves. We prove the following statements about trees
$T'_i$.

\begin{enumerate}
\item Each leaf of each $T'_i$ is incident to exactly one meaningful vertex.
\item For each edge $\{t_i,t_j\}$ of $U(H)$, there is exactly one meaningful vertex adjacent
to both $T'_i$ and $T'_j$.
\item If $t_i$ and $t_j$ are not adjacent in $U_H$ then there is no meaningful
vertex adjacent to both $T'_i$ and $T'_j$.
\end{enumerate}

Indeed, let $\ell_1$ be a leaf of $T'_i$. By definition of $I_i$,
$\ell_1$ is marked with a $j$ such that $\{t_i,t_j\}$ is an edge of $U(H)$.
Consequently, by definition of $I_j$, $T'_i$ has a leaf $\ell_2$ marked
with $i$. By definition of $G_m$, there exists a subdivision vertex $v$
incident to both $\ell_1$ and $\ell_2$. By definition of meaningful vertices,
$v$ is meaningful. Assume that $\ell_1$ is incident to another meaningful vertex $u \neq v$.
Then, by definition of $G_m$, $u$ connects $\ell_1$ to a leaf $\ell_3 \neq \ell_2$ of $T_j$.
By definition of a meaningful vertex, $\ell_3$ is marked with $i$ in contradiction to the
procedure of marking of $T^*_j$ that does not assign two different leaves with the same
element of $I_j$. Thus we have proved the first statement.

Let $\{t_i,t_j\}$ be an edge of $U(H)$. Then, by definition of $I_i$ and $I_j$,
$T'_i$ has a leaf $\ell_1$ marked with $j$ and $T'_j$ has a leaf $\ell_2$ marked with $i$.
The subdivision vertex $v$ connecting $\ell_1$ and $\ell_2$ is a meaningful vertex adjacent
to $T'_i$ and $T'_j$. Suppose there is another meaningful vertex $u$ adjacent to both
$T'_i$ and $T'_j$. Then $u$ is a subdivision vertex of $G_m$ adjacent to a leaf $\ell_3$
of $T'_i$ and a leaf $\ell_4$ of $T'_j$. Since $u \neq v$, by construction of 
$G_m$, either $\ell_3 \neq \ell_1$
or $\ell_4 \neq \ell_2$. Assume the former w.l.o.g. Then $\ell_3$ cannot be marked with
$j$ and hence $u$ is not meaningful, a contradiction. Thus we have proved the second statement.

For the third statement, assume by contradiction that there are $T'_i$ and $T'_j$ 
such that, on the one hand $\{t_i,t_j\}$ is not an edge of $U(H)$ and, on the other hand,
there is a meaningful vertex $v$ adjacent to both $T'_i$ and $T'_j$.
Let $\ell_1$ be the leaf of $T'_i$ adjacent to $v$. Then $\ell_1$ is marked with $j$
implying that $j \in I_i$, a contradiction proving the third statement.

For $1 \leq i \leq m$, let $T''_i$ be the rooted tree having the same root $t'_i$ as $T'_i$ 
and obtained from $T'_i$ by adding to each leaf $\ell$ of $T'_i$ the edge $\{\ell,v\}$ where
$v$ is the meaningful vertex adjacent to $\ell$ according to the first statement above.
Let $H'$ be the union of $T'_1, \dots, T'_m$. We are going to prove that 
$H' \in \btb$ with ${\bf T}(H')=\{T''_1, \dots, T''_m\}$ with $U(H')$ isomorphic
to $U(H)$. 

That $T''_1, \dots, T''_m$ are all extended follows by construction.
Let $v$ be a leaf of $T''_i$ and let us show that $v$ is a leaf of exactly one other
tree $T''_j$.  Let $\ell_1$ be the only neighbour of $v$ in $T''_i$.
By definition, $\ell_1$ is a leaf of $T'_i$. Let $j$ be the mark of $\ell_1$.
By definition of $v$ as a meaningful vertex, it is connected to a leaf $\ell_2$
of $T_j$ marked with $i$. We claim that $v$ is a leaf of $T_j$. Indeed, let $u$
be the leaf of $T''_j$ adjacent to $\ell_2$ (such a leaf exists by construction).
Then $u$ is a meaningful vertex adjacent to $\ell_2$. By the first statement above,
there is only one meaningful vertex adjacent to $\ell_2$ and hence $u=v$.
Note that $T''_i$ and $T''_j$ can only have joint leaves as $T'_i$ and $T'_j$ are
vertex disjoint. Moreover, by construction, a joint leaf of $T''_i$ and $T''_j$
is a meaningful vertex adjacent to both $T'_i$ and $T'_j$. By the third statement,
there can be at most one such a meaningful vertex, confirming that $H' \in \btb$.
The third statement in fact claims that such a vertex exists if and only $\{t_i,t_j\}$
is an edge of $U(H)$. That is $T''_i$ and $T''_j$ have a joint leaf if and only if 
$\{t_i,t_j\}$ is an edge of $U(H)$ establishing a natural isomorphism between
$U(H')$ and $U(H)$ with $t'_i$ corresponding to $t_i$.

Note that that the height of each $T'_i$ does not exceed the height of $T^*_i$
and hence is at most $\log m/4.9+2$. Hence, the height of each $T''_i$ is at most
$\log m/4.9+3$. Consequently, $H'$ is a pseudoexpander. 

It remains to show that $\varphi(G_m)|_S=phi(H')$.
By definition, the clauses of $\phi(H')$ correspond to the edges $U(H')$.
Consider the clause $C_{i,j}$ corresponding to an edge $\{t'_i,t'_j\}$ of $H'$.
Let $v$ be the joint leaf of $T''_i$ and $T''_j$.
By definition, $C_{i,j}$ contains the vertices of the root-leaf path
of $T''_i$ ending with $v$ and the vertices of the root-leaf path of $T''_j$ ending with $v$.
Recall that $v$ is a subdivision vertex of $G_m$ and note that $C_{i,j}$ is exactly the
clause $C_v$ of $\varphi(G_m)$. As $v$ is a meaningful vertex, it does not belong to
$S$ and hence $C_v$ is a clause of $\varphi(G_m)|_S$.

Conversely, let $C_v$ be a clause of $\varphi(G_m)|_S$. 
Then $v$ is a meaningful vertex. That is, $v$ is adjacent to respective leaves 
$\ell_1$ and $\ell_2$ of some trees $T'_i$ and $T'_j$. By construction and the
second statement above, $v$ is the joint leaf of $T''_i$ and $T''_j$ and, 
by definition of $C_v$, it is exactly $C_{i,j}$.
$\blacksquare$
 
\subsection{Proof of Theorem \ref{nbpfinal}}
Let $S$ be an assignment to the variables of $\varphi(G_m)$. Then for any 
fixed $d$, the size of smallest $d$-{\sc nbp} computing $\varphi(G_m)$ is greater
than or equal to the size of the smallest $d$-{\sc nbp} computing
$\varphi(G_m)|_S$. The reason for this is that an $d$-{\sc nbp} $Z$ computing
$\varphi(G_m)$ can be transformed into a $d$-{\sc nbp} computing $\varphi(G_m)|_S$
without increase of size. Indeed, for each edge $e$ of $Z$ labelled with a literal
$x$ of a variable of $S$, remove the label if $x \in S$ or remove the edge $e$ if
$\neg x \in S$. Let $Z'$ be an {\sc nbp} obtained by the union of all source-sink
paths of of $Z$ where no edge has been removed and the removal of labels as above
preserved. Then $Z'$ computes $\varphi(G_m)|_S$. Indeed, let $S'$ be such that
$S \cup S'$ is a satisfying assignment of $\varphi(G_m)$. Let $P$ be a computational path
of $Z$ with $A(P) \subseteq S \cup S'$. No edges of $P$ are removed and all the labels
of $S$ are removed in $Z'$ so $P$ remains a computational path of $Z'$ 
and the assignment on $P$ in $Z'$ is a subset of $S'$, hence
$S'$ is a satisfying assignment of the function computed by $Z'$. 

Conversely, let $S'$ be a satisfying assignment of the function computed by $Z'$.
That is $Z'$ has a computational path $P'$ with $A(P') \subseteq S'$. Then $P'$
is a path of $Z$ and additional labels on $P'$ are all elements of $S$
(they do not contain occurrences opposite to $S'$ and do not contain two opposite occurrences).
That is $P'$ is a computational path of $Z$ with an assignment being a subset of $S \cup S'$
and hence $S'$ is a satisfying assignment of $\varphi(G_m)$.

By Theorem \ref{peinfprel}, for each sufficiently large $m$
there is a pseudoexapnder $H_m$ with $m$ roots and max-degree of the underlying graph at
most $m^{1/4.9}$. Then it follows from Lemma \ref{lem:redpe} that for each sufficiently
large $m$ there is an assignment $S_m$ to the variables of $\varphi(G_m)$ such that
$\varphi(G_m)|_{S_m}=\phi(H'_m)$, where $H'_m$ is a pseudoexpander with $m$ roots. 
By Theorem \ref{lbmain} computing $\varphi(G_m)|_{S_m}$ requires $d$-{\sc nbp} of exponential
size in $m$. Since, as shown above, the size of $d$-{\sc nbp} needed to compute 
$\varphi(G_m)$ is at least as large and the number of variables of $\varphi(G_m)$
is polynomial in $m$, a $d$-{\sc nbp} of exponential size is required to compute $\varphi(G_m)$.
$\blacksquare$

\section{Proof of auxiliary lemmas for Theorem \ref{lbcons}}
{\bf Proof of Lemma \ref{lem:exsep}.}
We are going to prove the following statement. Let $P$ be a computational path 
of $Z$ on which each variable occurs exactly $d$ times where $d \leq \log m/40000$.
Then $P$ contains a set $X$ of vertices of size at most $\log m/4000$ that separates
two subsets of $Roots(H)$ of size at least $m^{0.999}/3$.

Let us verify that the lemma follows from the above statement.
By definition of $H$, $n \leq 3m^2$. Therefore, $\log n/10^5 \leq \log 3m^2/10^5=(\log m/50000)+3 \leq
\log m/40000$ for sufficiently large $m$.
Therefore, the above mentioned set $X$ of vertices can be found for $d \leq \log n/10^5$
and as $m \leq n$, the size of this set is at most $\log n/4000$, as required by the lemma.

{\bf Definition of an interval.}
For an arbitrary sequence $s_1, \dots, s_b$, and $1 \leq i \leq q \leq b$,
let us call the subsequence $s_i,s_{i+1}, \dots s_q$ of consecutive elements
of $s_1, \dots, s_b$ an \emph{interval} of $s_1, \dots, s_b$.

The key observation for this proof is the following claim.
\begin{claim} \label{cl:meander}
Let $SEQ$ be a sequence of elements of $\{1, \dots, m\}$ where each element appears
at most $m/40000$ times. Then, for a sufficiently large $m$, there is a partition
of $SEQ$ into intervals $SEQ_1, \dots, SEQ_{c+1}$, where $c \leq \log m/4000$
and two disjoint subsets $U_1$ and $U_2$ of size at most $m^{0.999}$ each such that one 
of the following two statements is true.
\begin{enumerate}
\item Elements of $U_1$ occur only in $SEQ_i$ with an odd $i$ and elements of $U_2$
occur only in $SEQ_i$ with an even $i$.
\item Elements of $U_1$ occur only in $SEQ_i$ with an even $i$ and elements of $U_2$
occur only in $SEQ_i$ with an odd $i$.
\end{enumerate}
\end{claim}

Let us show how the lemma follows from Claim \ref{cl:meander}.
Let $P$ be a computational path of $Z$.
Let $u_1, \dots, u_a$ be the sequence of literals of variables of ${\bf Roots}(H)$
appearing in the order as they occur along $P$.
Let $SEQ^*=v_1, \dots, v_q$ be the sequence where each $v_i$ is $u_i$ is $u_i$ is the positive
literal and $\neg u_i$ otherwise. 
That is, $SEQ^*$ is nothing else than a sequence of elements of ${\bf Roots}(H)$.
As $|{\bf Roots}(H)|=m$ and each element occurs at most $\log m/40000$ times (
see the second paragraph of the proof for justification), Claim \ref{cl:meander} applies to $SEQ^*$.
In particular, there is a partition $SEQ^*_1, \dots, SEQ^*_{c+1}$ of $SEQ^*$ into intervals,
where $c \leq \log m/4000$ and two disjoint subsets $U_1$ and $U_2$ of ${\bf Roots}(H)$ such that
at least one of the two statements of Claim \ref{cl:meander} happens with $SEQ^*$ replacing $SEQ$. 

Let $i_1, \dots, i_c$ be such that for $ 1 \leq j \leq c$, $v_{i_j}$ is the last element of $SEQ^*_j$.
It is not hard to see that $P$ can be partitioned into subpaths $P_1, \dots, P_{c+1}$
such that the sequence of occurrences of ${\bf Roots}(H)$ on $P_1$ is $u_1, \dots, u_{i_1}$ and for
each $1 < j \leq c+1$, the sequence of occurrences of ${\bf Roots}(H)$ on $P_j$ is
$u_{i_{j-1}+1}, \dots, u_{i_j}$.
Consequently one of the following two statements is true.
\begin{enumerate}
\item Elements of $U_1$ occur only in $P_i$ with an odd $i$ and elements of $U_2$
occur only in $P_i$ with an even $i$.
\item Elements of $U_1$ occur only in $P_i$ with an even $i$ and elements of $U_2$
occur only in $P_i$ with an odd $i$.
\end{enumerate}
Therefore, the set $x_1, \dots, x_c$ of respective ends of $P_1, \dots, P_{c+1}$
separates $U_1$ and $U_2$ as required.

{\bf Proof of Claim \ref{cl:meander}.}
An interval $S$ of $SEQ$ is a \emph{link} between two disjoint subsets $U_1$ and $U_2$
of $\{1, \dots, m\}$ if no element of $U_1 \cup U_2$ occurs as an intermediate element
of $S$ and one of the following two statements is true.
\begin{itemize}
\item An element of $U_1$ occurs as the first element of $S$ and an element of $U_2$
occurs as the last element of $S$.
\item An element of $U_1$ occurs as the first element of $S$ and an element of $U_2$
occurs as the last element of $S$.
\end{itemize}

Let $s=\log m/4000+1$. We claim that there are two disjoint subsets $U_1$ and $U_2$ 
of size $\ell=\lfloor m/(2^{s+1})\rfloor$
that $SEQ$ has at most $s-1$ links between them. Indeed, assume the opposite.
Then, in particular, $SEQ$ has at least $s$ links between every 
$U_1 \subseteq \{1, \dots, m/2\}$ and $U_2 \subseteq \{m/2+1, \dots, m\}$
of size $\ell$ each. By Theorem 1.1. of Alon and Maass \cite{AMaass},
$|SEQ| \geq 1/8*m(s-9)$. On the other hand, as each element of $\{1, \dots, m\}$
occurs at most $\log m/40000$ times, $|SEQ| \leq m*\log m/40000$. 
That is, $\log m/40000 \geq 1/8*(s-9)$ or, after a transformation,
$\log m/5000+9 \geq s=\log m/4000+1$, which is incorrect for a sufficiently large $m$.
This contradiction proves the correctness of the initial claim.

Note that $\ell \geq m/2^{\log m/4000+2}=m^{3999/4000}/4 \geq m^{0.999}$ for a sufficiently
large $m$. Therefore, it remains to show that there is a partition $SEQ_1, \dots, SEQ_{c+1}$
of $SEQ$ into intervals with $c \leq s-1$ for which one of the statements of Claim \ref{cl:meander} is true.
Let $I_1, \dots, I_c$ be a largest set of links between $U_1$ and $U_2$, $c \leq s-1$ by
assumption. 

As the first element of the link, completely determines the link itself, the first elements
of $I_1, \dots, I_c$ are all different. We assume that they occur on $SEQ$ in the order the
respective links are listed. Note that for two consecutive links $I_i,I_{i+1}$,
the first element of $I_{i+1}$ is either the last element of $I_i$ or occurs on $SEQ$
after the last element of $I_i$. Indeed, since the first element of $I_{i+1}$ occurs after 
the first element of $I_i$, any `deeper' overlap would imply that the first element of $I_{i+1}$
does not belong to $U_1 \cup U_2$, a contradiction. With this in mind, we can define the following
intervals $SEQ_1, \dots, SEQ_{c+1}$. $SEQ_1$ is the prefix of $SEQ$ whose last element is the first
element of $SEQ$. For each $1<i \leq c$, $SEQ_i$ is the interval of $SEQ$ whose first element 
is the one immediately following the last element of $SEQ_{i-1}$ and whose last element is the first
element of $I_i$. Finally, $SEQ_{c+1}$ is the suffix of $SEQ$ whose first element is the one
immediately following the last element of $SEQ_c$. Note that for each $1 \leq i \leq c$, the last
element of $I_i$ belongs to $SEQ_{i+1}$.

Assume w.l.o.g. that the first element of $I_1$ belongs to $U_1$. Then we prove by induction
on $1 \leq i \leq c+1$ That elements of $U_1$ occur in only in intervals $SEQ_i$ with $i$
being odd and the elements of $U_2$ occur only in intervals $SEQ_i$ with $i$ being even.

Consider first $SEQ_1$. By assumption, its last element belongs to $U_1$.
Assume that $SEQ_1$ has another element of belongs to $U_2$. As all the elements of $U_2$
occur in $SEQ_1$ before the last one, we can identify the last element $u$ in $SEQ_1$ that
belongs to $U_2$ and the first element $v$ of $U_1$ following $u$. Then the interval between
$u$ and $v$ is a link located before $I_1$, that is $SEQ$ has at least $c+1$ links between
$U_1$ and $U_2$ in contradiction to the maximality of $c$. 

Assume now that $i>1$.
Assume first that $i \leq c$.
We can assume w.l.o.g. that $i$ is even for the proof for the case where $i$ is odd
is symmetric. By the induction assumption, all the elements of $U_1 \cup U_2$
that belong to $SEQ_{i-1}$ in fact belong to $U_1$. In particular, the first element
of $I_{i-1}$ belongs to $U_1$. Therefore, the last element $u'$ of $I_{i-1}$ belongs
to $U_2$. By construction, any element of $SEQ_i$ preceding $u'$ is an intermediate element
of $I_{i-1}$ and hence does not belong to $U_1 \cup U_2$. Therefore, if $SEQ_i$ contains 
an element of $U_1$ then it occurs after $u'$. Consequently, we can identify on
$SEQ_i$ the first element $v$ of $U_1$ and the last element $u$ of $U_2$ preceding it.
The interval between $u$ and $v$ is a link that lies between $I_{i-1}$ and $I_i$,
again implying that $SEQ$ has at least $c+1$ links in contradiction to the maximality of $c$.

Assume now that $i=c+1$. Again,  we can assume w.l.o.g. that $i$ is even.
Arguing as in the previous case, we observe that if $SEQ_{c+1}$ contains elements of
$U_1$ then $SEQ_{i+1}$ contains a link occurring after $I_c$ in contradiction to the
maximality of $c$.
$\blacksquare$

{\bf Proof of Lemma \ref{lem:allneg}}
Let $c$ be such that $|X|=c-1$ and let $Q$ be a path
containing $X$ and on which $X$ generates a partition 
$Q_1, \dots, Q_c$ such that one of the following two statements is true.
\begin{enumerate}
\item Elements of $Y_1$ occur only on $Q_i$ with an odd $i$ and elements of $Y_2$ occur on only $Q_i$ with an even $i$.
\item Elements of $Y_1$ occur only on $Q_i$ with an even $i$ and elements of $Y_2$ occur on only $Q_i$ with an odd $i$.
\end{enumerate}
We assume w.l.o.g. that the first statement is true.
In fact, due to the uniformity of $Z$, a stronger statement is true.

\begin{claim} \label{cl:unicon}
For each path $P$ containing $X$, let $P_1, \dots, P_c$ be the partition of $P$ generated 
by $X$. Then elements of $Y_1$ occur only on $P_i$ with an odd $i$ and elements of $Y_2$ occur on only 
$P_i$ with an even $i$.
\end{claim}

{\bf Proof.}
Due to $Z$ being a DAG, the vertices of $X$ occur in the same order on each path containing $X$.
In particular, this means that for each $1 \leq i \leq q$, $P_i$ and $Q_i$ have the same
initial vertex and the same final vertex. Assume that the claim does not hold for
some $1 \leq i \leq q$. If $i$ is odd this means that an element of $Y_2$ occurs in $P_i$ 
and does not occur in $Q_i$ in contradiction to Lemma \ref{lem:samevar}.
Similarly, if we assume that the claim does not hold for some even $i$ then this means that 
an element of $Y_1$ occurs on $P_i$ while not occurring on $Q_i$ again in contradiction to
Lemma \ref{lem:samevar}.
$\square$

Let $Q'_1, \dots, Q'_c$ and $Q''_c \dots Q''_c$ be respective partitions of $Q'$ and
$Q''$ generated by $X$.
Let $Q^*_k=Q'_k$ whenever $k$ is odd and $Q^*_k=Q''_k$
whenever $k$ is even. Let $Q^*=Q^*_1+ \dots +Q^*_c$. Clearly, $Q^*$ is a 
source-sink path of $Z$. We prove that $Q^*$ is a computational path assigning all
the vertices of $Y'_1 \cup Y'_2$ negatively. 

Assume the first statement does not hold, that is $Q^*$ is not a computational path.
Then there is a variable $x$ that occurs on $Q^*$ both positively and negatively.
As all the labels of $Q^*$ are either labels of $Q'$ or labels of $Q''$ which
do not have opposite variable occurrences due to being computational paths,
variable $x$ occurs positively in one of $A(Q''), A(Q'')$ and negatively in the other.
As all the variables of $Vars(Z) \setminus (Y_1 \cup Y_2)$ have the same occurrence in both
$A(Q')$ and $A(Q'')$ $x$ must belong to $Y_1 \cup Y_2$. 

Assume w.l.o.g. that $x \in Y_1$. The case with $x \in Y_2$ is symmetric.
By Claim \ref{cl:unicon} $x$ occurs only on paths $Q^*_i$ with odd $i$.
However, all these paths are subpaths of $Q'$. That is, all the occurrences of $x$
on $Q^*$ are occurrences of $x$ on $Q'$. Since the latter is a computational
path and cannot contain opposite occurrences of the same variable,
$Q^*$ does not contain opposite occurrences of $x$. Thus we have arrived at a 
contradiction that the there are no variables having opposite occurrences on $Q^*$,
thus confirming that $Q^*$ is a computational path.

In order to show that all the variables of $Y'_1 \cup Y'_2$ are falsified by $Q^*$, 
consider a variable $x \in Y'_1$. Then all the occurrences of $x$ are on $Q^*_i$
with an odd $i$. That is, all the occurrences of $x$ on $Q^*$ are occurrences of $Q'$.
As $Q'$ assigns $x$ negatively, so is $Q^*$. For $x \in Y'_2$, the reasoning is symmetric
with the even fragments of $Q^*$ used instead the odd ones and $Q''$ instead $Q'$.
$\blacksquare$

\section{Proof of Proposition \ref{prop:validsp} (Sketch).}
Let us pick a random assignment to the variables of $H$ using the following procedure.

\begin{enumerate}
\item Arbitrarily order the variables of $\phi_H$ so that all the non-leaf variables
are ordered before the leaf variables.
\item Pick a literal of each non-leaf variable with probability $0.5$ for both positive and negative literals.
\item For each leaf variable $\leaf_{i,j}$, do as follows.
   \begin{itemize}
   \item If a positive literal has been chosen for at least one non-leaf variable of $C_{i,j}$
         choose either the positive or the negative literal of $\leaf_{i,j}$ with probability $0.5$.
   \item Otherwise, choose the positive literal of $\leaf_{i,j}$ with probability $1$. 
   \end{itemize}
\end{enumerate}

Let the probability of the resulting assignment be the product of the individual probabilities
of the chosen literals. 

Let $S \in {\bf SAT}$. By construction, the probabilities of literals of non-leaf variables are $0.5$
and the probabilities of literals of leaf variables are $0.5$ if and only if they are not fixed.
It follows that the probability assigned to $S$ by the above procedure is $(1/2)^{|S \setminus Fix(S)|}$
exactly as in the defined probability space.

If $S \notin {\bf SAT}$ then there is a clause $C_{i,j}$ falsified by $S$. In particular,
this means that $\leaf_{i,j}$ occurs negatively in $S$ as well as the rest of variables of
$C_{i,j}$. Clearly, according to the above procedure, the probability of a negative literal
of $\leaf_{i,j}$ is $0$ and hence the probability assigned to $S$ is $0$ as well.

It follows that the sum of probabilities of assignments obtained by the above procedure 
is in fact the sum of probabilities of satisfying assignments (which are the same as in our 
probability space). To show that the sum of probabilities is $1$, represent the process of choosing
a random assignment as a rooted decision tree with probabilities of literals being the weights of the edges of the
tree and the weight of a path being the product of weights of its edges and the weight of a collection
of paths being the sum of weights of paths in this collection. Then, starting from the leaves
and moving towards the root, argue by induction that the weight of each subtree is $1$.
$\blacksquare$


\end{document}